\documentclass[runningheads]{llncs}
\usepackage[T1]{fontenc}
\usepackage{graphicx}
\usepackage{float}
\usepackage{bbm}
\usepackage{stmaryrd}
\usepackage{graphicx,wrapfig,lipsum}
\usepackage{mathtools}
\usepackage{wasysym}
\usepackage{xcolor}
\usepackage{bussproofs}
\usepackage[hypertexnames=false]{hyperref}
\usepackage{stmaryrd}
\usepackage{amsfonts}
\usepackage{amssymb}
\usepackage{mathpartir}
\usepackage{mathrsfs}
\usepackage{tikz,pgfplots}
\usepackage{enumitem}
\usepackage{orcidlink}
\usepackage{url}
\usepackage{listings}
\definecolor{kwpurple}{RGB}{127,0,204}   
\definecolor{typegreen}{RGB}{0,140,0}    
\definecolor{strred}{RGB}{200,0,0}       
\definecolor{lnum}{RGB}{60,60,60}        
\definecolor{graymeta}{RGB}{90,90,90}    
\lstdefinelanguage{mydsl}{%
  sensitive=true,
  keywords={let,rec,node,in,if,then,else,fby,models,exec,reset,init,der},
  keywordstyle=\color{kwpurple}\bfseries,
  emph={float,int,bool,real},
  emphstyle=\color{typegreen}\bfseries,
  morestring=[b]",
  stringstyle=\color{strred},
  morecomment=[l]{--},
  commentstyle=\color{lnum}\itshape
}

\newif\ifdraft
\draftfalse

\newcommand{\jiawei}[1]{\ifdraft\textsf{\color{purple}[#1 --Jiawei]}\fi}

\newcommand{\kw}[1]{{\sf #1}}
\pgfplotsset{compat=1.18}
\begin{document}
\title{Towards Formal Verification of Hybrid Synchronous Programs with Refinement Types}

%
\author{Serra Z. Dane\inst{1}\orcidlink{0009-0008-9028-6597} \and
Jiawei Chen\inst{1}\orcidlink{0000-0002-5461-6711} \and Marc Pouzet\inst{2}\orcidlink{0000-0002-2651-7708} \and
Jean-Baptiste Jeannin\inst{1}\orcidlink{0000-0001-6378-1447}}
\authorrunning{Dane et al.}
\titlerunning{Towards Formal Verification of Hybrid Synchronous Programs}
%
\institute{University of Michigan, Ann Arbor, MI, USA \\ \email{\{sdane, chenjw, jeannin\}@umich.edu}\and ENS, PSL University and INRIA, Paris, France\\ \email{marc.pouzet@ens.fr}}
%
\maketitle              
\begin{abstract}
Cyber-physical systems (CPS) such as autonomous cars, aircraft, and robots are often also safety-critical; thus it is imperative that they operate as intended with a high degree of certainty. Formal verification has been employed to verify the software controlling these systems, but due to their complexity, is usually performed on an abstract model rather than the executable code. Synchronous programming languages extended with differential equations promise both rigorous modeling and sufficient expressiveness to implement executable controller code, and recent developments have introduced formal verification of strictly discrete-time programs. Extending these verification techniques to hybrid systems enables precise modeling of the environment for a wider variety of programs to be both verified and executed. We formalize the operational semantics of initial value problems and zero-crossing detection expressed in a synchronous programming language, extend its type system for verification thereof, and prove its soundness.


\end{abstract}

\section{Introduction}
Cyber-physical systems (CPS), systems in which software operates in a physical environment, are fundamental for modeling real-world systems such as robots and embedded systems. Because of their physical effects, the software running these systems is often safety-critical. Rigorous techniques such as formal methods are preferred over testing due to the complex interactions present. While formal verification is typically done on an abstract mathematical model of the system, direct verification of the executable code controlling a CPS provides a more convincing result. This necessitates the development of semantics that can accurately characterize the program. However, this is complicated by the fact that CPS are often hybrid: the controller runs in discrete, synchronous steps, while the physical environment evolves continuously (often modeled by ODEs), so the overall behavior mixes discrete updates with continuous-time dynamics. A promising direction is therefore to bring verification and implementation closer together in a single programming language with precise hybrid semantics, so that verified properties directly correspond to the behavior of executable controller code. The programming language for writing models must be mathematically precise (that is, with unambiguous semantics and implementation) such that all validation steps can be performed on the source, to give confidence on the model itself and to automatically transfer some of the properties to the implementation.

Synchronous programming languages provide a strong foundation for this goal such as Zélus~\cite{bourke2013zelus}, Lustre~\cite{caspi1987lustre}, Esterel~\cite{berry1992esterel}, Signal~\cite{leguernic1991signal} and SCADE~\cite{colaco2017scade6}. Their stream-based, deterministic model of computation matches the structure of many control programs, and their discipline of logical time aligns well with embedded execution. Moreover, modern synchronous languages such as Zélus extend this discrete-time core with continuous-time dynamics expressed as ordinary differential equations (ODEs), enabling programmers to model hybrid CPS in a unified setting where discrete control logic coexists with continuous evolution \cite{bourke2013zelus}.  In such hybrid synchronous programs, the connection between discrete and continuous behavior is often mediated by zero-crossings: discrete events occur when a guard function over the continuous state crosses zero, triggering resets or mode changes. Zero-crossings are thus central to both modeling and execution.

Recent work has shown that refinement types can be used to verify safety properties of discrete-time synchronous programs directly in the source language. Invariants are encoded as type refinements and discharging the resulting proof obligations with SMT solving \cite{chen_synchronous_2024}. However, extending this approach from discrete-time streams to hybrid programs raises fundamental semantic and proof-theoretic challenges. Continuous behavior exists outside of discrete-time computation so safety must be shown for both continuous and discrete time, and zero-crossings need to be defined precisely enough to support a sound theory. At the same time, continuous-time verification benefits from proof principles that differ from purely discrete induction; in particular, differential invariants, as developed in differential dynamic logic, provide a well-established way to prove that a property is preserved along ODE solutions \cite{platzer2008dl}. 

 This paper develops a refinement-type-based verification framework for hybrid synchronous programs that supports execution while providing formal guarantees against safety specifications. We present the following contributions:
\begin{enumerate}
    \item A formal, modular definition of zero-crossings for hybrid synchronous programs.
    \item Operational semantics and typing rules for verifying a synchronous programming language with continuous dynamics and event-triggered resets based on the hybrid synchronous language Zélus.
    \item A definition and proof of type safety for the hybrid synchronous refinement type metatheory.
\end{enumerate}

\section{Preliminaries}\label{sec:prelim}
In this section, we present the preliminaries required for the development of our theory.
\subsection{Synchronous Programming}
Synchronous languages model computations as \emph{streams} whose values evolve over
\emph{logical time}~\cite{caspi1987lustre,berry1992esterel,leguernic1991signal,colaco2017scade6}. A program reacts in discrete logical instants (clock cycles), producing a value for
each stream at every instant, and all computations within an instant must complete before
the next instant begins. This yields a deterministic model of computation that
is well suited to control software and embedded systems where predictability is essential.
Under the \emph{synchronous hypothesis}, each reaction is treated as instantaneous at the level
of logical time, providing deterministic interaction points with the environment~\cite{berry1989realtime}.
This discipline also supports modular compilation and compositional reasoning about stream-based
controllers~\cite{benveniste2003synchronous12}. These properties make synchronous programming a
natural foundation for hybrid systems: discrete controller updates occur at logical instants, while
continuous plant evolution (and event detection) can be analyzed over the real-time intervals between
reactions.

\subsection{Hybrid Systems}
Hybrid systems combine \emph{discrete} control with \emph{continuous} physical dynamics. A standard
mathematical model is the \emph{hybrid automaton}: a finite set of modes, an ODE (flow) associated with
each mode, and guarded discrete transitions that may reset the continuous state
\cite{alur1995hybrid,henzinger1996theory}. Executions alternate between continuous evolution that
follows the active mode’s ODE and instantaneous jumps that change mode and/or apply a reset.
This model is well suited to CPS: controllers update at discrete instants (e.g., periodic sampling
and actuation), while the plant evolves in continuous time between updates. In this work we adopt a
hybrid \emph{synchronous} programming perspective via Z\'elus, which extends a Lustre-style synchronous
language with ODE-based continuous evolution and event-triggered resets, allowing discrete control
and continuous dynamics to be expressed and executed within a single source language
\cite{benveniste_hybrid_2011,bourke2013zelus}.

\subsection{Differential Dynamic Logic (dL)}
Our type system relies on proof rules for reasoning about continuous evolution in the style of
Platzer's \emph{differential dynamic logic} (dL) \cite{platzer2018lfcps}. In dL, an ODE evolution has the form
\[
  \dot x = f(x)\ \&\ Q(x),
\]
where $\dot x$ denotes the time derivative of the state $x(t)$ along a trajectory, i.e.,
$\dot x(t)=\frac{d}{dt}x(t)$. Here the state variable $x$ evolves continuously according to the vector field $f$ for an arbitrary duration, \emph{restricted} to times during which the evolution-domain constraint $Q$ holds.
Modal formulas $[\alpha]\phi$ (resp. $\langle\alpha\rangle\phi$) state that $\phi$ holds after all (resp. some) terminating runs of program $\alpha$; we consider only the former case.

A central proof technique we adopt is the \emph{differential invariant} rule. Informally, to prove that a predicate
$P$ holds throughout an ODE evolution restricted to $Q$, it suffices to show that $P$ is preserved differentially while $Q$ holds:
\begin{mathpar}
\inferrule*[right=(dI)]
{
  Q \vdash [ \dot x:=f(x)]\dot P
}
{
  P \vdash [\dot x = f(x)\ \&\ Q\,]\ P
}
\end{mathpar}
Here $[\,\alpha\,]\phi$ means $\phi$ holds after all terminating runs of $\alpha$, and
$\dot x=f(x)\ \&\ Q$ denotes continuous evolution $\dot x=f(x)$ restricted to states satisfying $Q$.
The formula $\dot P$ is the symbolic time derivative of $P$ along the flow, and $\dot x:=f(x)$ assigns the
derivative symbol $\dot x$ to $f(x)$ to reduce checking $\dot P$ to an algebraic condition under $Q$.
Note that this differential evaluation can stop even though $Q$ does not become $false$, since in dL an ODE command is non-deterministic in its duration~\cite{platzer2018lfcps}. Unlike dL, where $Q$ explicitly restricts the ODE and there is no determined duration of evaluation, adapting this rule to zero-crossings requires additional safeguards, established in subsequent sections.

\subsection{Refinement Types}
Refinement types strengthen a base type \(b\) with a logical predicate \(\varphi\), written
\(\{v:b \mid \varphi(v)\}\), to specify subsets of values of \(b\) satisfying \(\varphi\). For example, positive reals can be written \(\{v:\mathsf{real}\mid v>0\}\)~\cite{jhala2020refinement}. Operationally, refinement typing supports generating verification conditions that can be discharged
by automated solvers, while keeping the core language type system largely unchanged \cite{rondon2008liquid}. We use refinements to
express \emph{safety} properties of hybrid state variables and to encode inductive invariants that must
hold across both continuous evolution and discrete resets.

\subsection{Verification Via Typing}
Verification via typing uses a type system as a lightweight, compositional proof discipline: a program that type-checks is guaranteed to satisfy a class of semantic safety properties, with the guarantee justified by a \emph{type soundness} theorem relating typing to the language’s operational semantics \cite{wright1994syntactic}. Liquid Types combine Hindley--Milner inference with predicate abstraction to automatically infer refinements strong enough to prove program invariants \cite{rondon2008liquid}. Other approaches internalize program logics into types: Hoare Type Theory integrates Hoare-style pre/postconditions into the type of effectful computations, enabling modular reasoning about stateful code within the type system \cite{nanevski2008hoare}. Typing-based verification has also been explored in the context of synchronous programming, where refinement types are used to relate executable models and safety specifications within a unified language and metatheory \cite{chen_synchronous_2024}. However, this line of work focuses on a restricted fragment of the language limited to purely synchronous constructs, leaving open the challenge of handling richer hybrid features addressed in the present work.

\section{Formal Definition of Zero-Crossings}\label{sec:zc}
In this section we present a formal definition of \emph{zero-crossings} based on the observable
behavior of real-valued guard functions. Intuitively, a zero-crossing occurs when a function
$g(t)$ changes sign, i.e., it passes from negative to positive or vice versa. In hybrid systems modelers, however, such sign tests are typically not primitive semantic constructs. Conditions
like $x \ge 0$ are implemented via an internal mechanism known as \emph{zero-crossing detection}~\cite{cellier1991continuous,benveniste2012nonstandard,lee2005operational}.

There does not appear to be a single standard formal definition for zero-crossings: candidate definitions based solely on sign change disagree on subtle
behaviors near zero, leading to implementation-dependent behavior and ambiguity in formal reasoning~\cite{benveniste2012nonstandard}.
Since Zélus depends on zero-crossings to trigger discrete computation, this definition is fundamental to all the subsequent formalism.


\paragraph{Notation.}
To make our results as general as possible, we introduce a new notation. For any $a\in\{-\infty\}\cup\mathbb{R}$ and $b\in\mathbb{R}\cup\{+\infty\}$, we write $\llparenthesis a;b \rrparenthesis$ for:
\begin{itemize}
\item $\llparenthesis a;b \rrparenthesis \triangleq [a;b]$ if $a\in\mathbb{R}$ and $b\in\mathbb{R}$;
\item $\llparenthesis a;b \rrparenthesis \triangleq (-\infty;b]$ if $a=-\infty$ and $b\in\mathbb{R}$;
\item $\llparenthesis a;b \rrparenthesis \triangleq [a;+\infty)$ if $a\in\mathbb{R}$ and $b=+\infty$;
\item $\llparenthesis a;b \rrparenthesis \triangleq (-\infty;+\infty)$ if $a=-\infty$ and $b=+\infty$.
\end{itemize}
Note that $\llparenthesis a;b \rrparenthesis\subseteq\mathbb{R}$ and cannot include $-\infty$ or $+\infty$.

\paragraph{Basic hypotheses and setup.}
Throughout this section, let $f$ be a function defined as {\bf continuous} on $\llparenthesis\ell; u\rrparenthesis$ with $\ell\in\{-\infty\}\cup\mathbb{R}$, $u\in\mathbb{R}\cup\{+\infty\}$ and $\ell\leq u$.
We consider only continuous functions.
Let $z\in\llparenthesis\ell; u\rrparenthesis$ such that $f(z)=0$.
The goal is to determine {\bf under which circumstances $z$ should be considered a zero-crossing.} We limit ourselves to negative-to-positive zero-crossings; the opposite case is symmetric.

We first note that the function $f$ might be identically zero on an interval of non-zero length around $z$. We would like to precisely define this segment, which we will write $\llparenthesis a; b\rrparenthesis.$ For this purpose, let us define $a$ and $b$ as:
\begin{itemize}
\item $a \triangleq \inf\{a\in\llparenthesis\ell; u\rrparenthesis\ |\ \forall x\in\llparenthesis a;z], f(x)=0\}$;
\item $b \triangleq \sup\{b\in\llparenthesis\ell; u\rrparenthesis\ |\ \forall x\in[z;b\rrparenthesis, f(x)=0\}$.
\end{itemize}
%

\subsection{Case Distinction}
We now have the following properties for $a$ and $b$:
\begin{itemize}
\item $a\in\{-\infty\}\cup\mathbb{R}$ and $b\in\mathbb{R}\cup\{+\infty\}$;
\item $\forall x\in \llparenthesis a; b\rrparenthesis, f(x) = 0$;
\item $\llparenthesis a; b\rrparenthesis \subseteq \llparenthesis l; u\rrparenthesis$, i.e.  $\ell \leq a \leq b \leq u$;
\item if $a>l\geq -\infty$, then $\forall \varepsilon>0, \exists x\in [a-\varepsilon;a), f(x)\neq 0$; 
\item if $b<u\leq +\infty$, then $\forall \varepsilon>0, \exists x\in (b;b+\varepsilon], f(x)\neq 0$.
\end{itemize}
If $a=b$, the function $f$ is only equal to zero at one point $z=a=b$, and non-zero right before and right after: we will say that the function is \emph{passing} through $z=a=b$. If $a<b$, the function $f$ is identically zero on $\llparenthesis a; b\rrparenthesis$, we will say that the function $f$ is $\emph{staying}$ on $\llparenthesis a; b\rrparenthesis$.



\newcommand\plotszone[2]{
\begin{tikzpicture}
\begin{axis}[xmin=-1.5,xmax=1.5, ymin=-1, ymax=1, axis x line=middle, axis y line=middle, 
                    xlabel={}, ylabel={}, ytick={0}, yticklabels={}, xtick={0}, xticklabels={},
                    every axis plot/.append style={smooth, very thick, no marks}]%
    \addplot[blue, restrict x to domain=-1.5:0, samples at={-1.5,-1.495,...,-0.01,-0.005,-0.0001}]{#1};
    \addplot[blue, restrict x to domain=0:1.5, samples at={0.0001,0.005,0.01,...,1.495,1.5}]{#2};
\end{axis} 
\end{tikzpicture}
}

\newcommand\plotsztwo[2]{
\begin{tikzpicture}
\begin{axis}[xmin=-1,xmax=2, ymin=-1, ymax=1, axis x line=middle, axis y line=middle, 
                    xlabel={}, ylabel={}, ytick={0}, yticklabels={}, xtick={0,1}, xticklabels={},
                    every axis plot/.append style={smooth, very thick, no marks}]
    \addplot[blue, restrict x to domain=-1:0, samples at={-1.0,-0.995,...,-0.01,-0.005,-0.0001}]{#1}; 
    \addplot[blue, restrict x to domain=0:1, samples at={0.0, 0.1,...,0.9,1.0}]{0};
    \addplot[blue, restrict x to domain=1:2, samples at={1.0001,1.005,1.01,...,1.995,2.0}]{#2};
\end{axis} 
\end{tikzpicture}
}

\newcommand\plotsz[4]{
\resizebox{1.95cm}{1.55cm}{ 
\begin{tabular}{c}
\plotszone{#1}{#2}
\\
\plotsztwo{#3}{#4}
\end{tabular}}
}

\newcommand\plotszNA[2]{
\resizebox{1.95cm}{1.55cm}{
\begin{tabular}{c}
\phantom{\plotsztwo{#1}{#2}}
\\
\plotsztwo{#1}{#2}
\end{tabular}}
}

\newcommand\plotszAlone[2]{
\resizebox{1.95cm}{1.55cm}{ 
\plotsztwo{#1}{#2}}
}

With this definition, we identify 85 cases arising from various combinations of behaviors to the left and right of a candidate zero-crossing, as well as passing and staying, summarized in the extended version of the paper and Figure~\ref{fig:zero-crossing-cases}.

\begin{figure}
    \centering
\resizebox{1\linewidth}{!}{

\begin{tabular}{c|c|c|c|c|c|c|c}
& 1. & 2. & 3. & 4. & 5. & 6. & 7.\\\hline
 A. &
 \plotsz{-x^2}{x^2}{-x^2}{(x-1)^2} &
 \plotsz{-x^2}{-x^2}{-x^2}{-(x-1)^2} &
 \plotsz{-x^2}{x*1.5*sin(deg(3/x))}{-x^2}{(x-1)*1.5*sin(deg(3/(x-1))} &
 \plotsz{-x^2}{x*1.5*(sin(deg(3/x)))^2}{-x^2}{(x-1)*1.5*(sin(deg(3/(x-1))))^2} &
 \plotsz{-x^2}{-x*1.5*(sin(deg(3/x)))^2}{-x^2}{-(x-1)*1.5*(sin(deg(3/(x-1))))^2} &
 \plotsz{-x^2}{5}{-x^2}{5} &
 \plotszNA{-x^2}{0}
\\\hline
 B. &
 \plotsz{x^2}{x^2}{x^2}{(x-1)^2} &
 \plotsz{x^2}{-x^2}{x^2}{-(x-1)^2} &
 \plotsz{x^2}{x*1.5*sin(deg(3/x))}{x^2}{(x-1)*1.5*sin(deg(3/(x-1))} &
 \plotsz{x^2}{x*1.5*(sin(deg(3/x)))^2}{x^2}{(x-1)*1.5*(sin(deg(3/(x-1))))^2} &
 \plotsz{x^2}{-x*1.5*(sin(deg(3/x)))^2}{x^2}{-(x-1)*1.5*(sin(deg(3/(x-1))))^2} &
 \plotsz{x^2}{5}{x^2}{5} &
 \plotszNA{x^2}{0}
\\\hline
 C. &
 \plotsz{-x*1.5*sin(deg(3/x))}{x^2}{-x*1.5*sin(deg(3/x))}{(x-1)^2} &
 \plotsz{-x*1.5*sin(deg(3/x))}{-x^2}{-x*1.5*sin(deg(3/x))}{-(x-1)^2} &
 \plotsz{-x*1.5*sin(deg(3/x))}{x*1.5*sin(deg(3/x))}{-x*1.5*sin(deg(3/x))}{(x-1)*1.5*sin(deg(3/(x-1)))} &
\plotsz{-x*1.5*sin(deg(3/x))}{x*1.5*(sin(deg(3/x)))^2}{-x*1.5*sin(deg(3/x))}{(x-1)*1.5*(sin(deg(3/(x-1))))^2} &
 \plotsz{-x*1.5*sin(deg(3/x))}{-x*1.5*(sin(deg(3/x)))^2}{-x*1.5*sin(deg(3/x))}{-(x-1)*1.5*(sin(deg(3/(x-1))))^2} &
 \plotsz{-x*1.5*sin(deg(3/x))}{5}{-x*1.5*sin(deg(3/x))}{5} &
 \plotszNA{-x*1.5*sin(deg(3/x))}{0}
\\\hline 
D. &
 \plotsz{x*1.5*(sin(deg(3/x)))^2}{x^2}{x*1.5*(sin(deg(3/x))^2}{(x-1)^2} &
 \plotsz{x*1.5*(sin(deg(3/x)))^2}{-x^2}{x*1.5*(sin(deg(3/x))^2}{-(x-1)^2} &
 \plotsz{x*1.5*(sin(deg(3/x)))^2}{x*1.5*sin(deg(3/x))}{x*1.5*(sin(deg(3/x))^2}{(x-1)*1.5*sin(deg(3/(x-1)))} &
\plotsz{x*1.5*(sin(deg(3/x)))^2}{x*1.5*(sin(deg(3/x)))^2}{x*1.5*(sin(deg(3/x))^2}{(x-1)*1.5*(sin(deg(3/(x-1))))^2} &
 \plotsz{x*1.5*(sin(deg(3/x)))^2}{-x*1.5*(sin(deg(3/x)))^2}{x*1.5*(sin(deg(3/x))^2}{-(x-1)*1.5*(sin(deg(3/(x-1))))^2} &
 \plotsz{x*1.5*(sin(deg(3/x)))^2}{5}{x*1.5*(sin(deg(3/x))^2}{5} &
 \plotszNA{x*1.5*(sin(deg(3/x)))^2}{0}
\\\hline 
E. &
 \plotsz{-x*1.5*(sin(deg(3/x)))^2}{x^2}{-x*1.5*(sin(deg(3/x)))^2}{(x-1)^2} &
 \plotsz{-x*1.5*(sin(deg(3/x)))^2}{-x^2}{-x*1.5*(sin(deg(3/x)))^2}{-(x-1)^2} &
 \plotsz{-x*1.5*(sin(deg(3/x)))^2}{x*1.5*sin(deg(3/x))}{-x*1.5*(sin(deg(3/x)))^2}{(x-1)*1.5*sin(deg(3/(x-1)))} &
\plotsz{-x*1.5*(sin(deg(3/x)))^2}{x*1.5*(sin(deg(3/x)))^2}{-x*1.5*(sin(deg(3/x)))^2}{(x-1)*1.5*(sin(deg(3/(x-1))))^2} &
 \plotsz{-x*1.5*(sin(deg(3/x)))^2}{-x*1.5*(sin(deg(3/x)))^2}{-x*1.5*(sin(deg(3/x)))^2}{-(x-1)*1.5*(sin(deg(3/(x-1))))^2} &
 \plotsz{-x*1.5*(sin(deg(3/x)))^2}{5}{-x*1.5*(sin(deg(3/x)))^2}{5} &
 \plotszNA{-x*1.5*(sin(deg(3/x)))^2}{0}
\\\hline
F. &
 \plotsz{5}{x^2}{5}{(x-1)^2} &
 \plotsz{5}{-x^2}{5}{-(x-1)^2} &
 \plotsz{5}{x*1.5*sin(deg(3/x))}{5}{(x-1)*1.5*sin(deg(3/(x-1))} &
 \plotsz{5}{x*1.5*(sin(deg(3/x)))^2}{5}{(x-1)*1.5*(sin(deg(3/(x-1))))^2} &
 \plotsz{5}{-x*1.5*(sin(deg(3/x)))^2}{5}{-(x-1)*1.5*(sin(deg(3/(x-1))))^2} &
\resizebox{1.95cm}{1.75cm}{
\begin{tabular}{c}
\begin{tikzpicture}
\begin{axis}[xmin=-1.5,xmax=1.5, ymin=-1, ymax=1, axis x line=middle, axis y line=middle, 
                    xlabel={}, ylabel={}, ytick={0}, yticklabels={}, xtick={0}, xticklabels={},
                    every axis plot/.append style={smooth, very thick, no marks}]
    \addplot[blue, restrict x to domain=-0.015:0.015, samples at={-0.015,0,0.015}]{0};
\end{axis} 
\end{tikzpicture}
\\
\plotsztwo{5}{5}
\end{tabular}}
 &
 \plotszNA{5}{0}
\\\hline G. &
 \plotszAlone{0}{(x-1)^2} &
 \plotszAlone{0}{-(x-1)^2} &
 \plotszAlone{0}{(x-1)*1.5*sin(deg(3/(x-1))} & 
 \plotszAlone{0}{(x-1)*1.5*(sin(deg(3/(x-1))))^2} & 
 \plotszAlone{0}{-(x-1)*1.5*(sin(deg(3/(x-1))))^2} &
 \plotszAlone{0}{5} & 
 \plotszAlone{0}{0}
\end{tabular}


}
    \caption{Visualization of the 85 different cases studied in this paper: Rows A–G represent the seven possible behaviors of 
\emph{f} on arbitrarily small left neighborhoods of \emph{a}, and columns 1–7 represent the symmetric seven behaviors on arbitrarily small right neighborhoods of \emph{b}. Each panel shows the function shape corresponding to the intersection of the row and column cases. For each case (e.g. E4), there are two subcases, for the passing case ($a=b$) and the staying case ($a<b$). All cases are explained in detail in the extended version of the paper. For cases in line G and column 7, only the staying case is applicable.}
    \label{fig:zero-crossing-cases}
\end{figure}
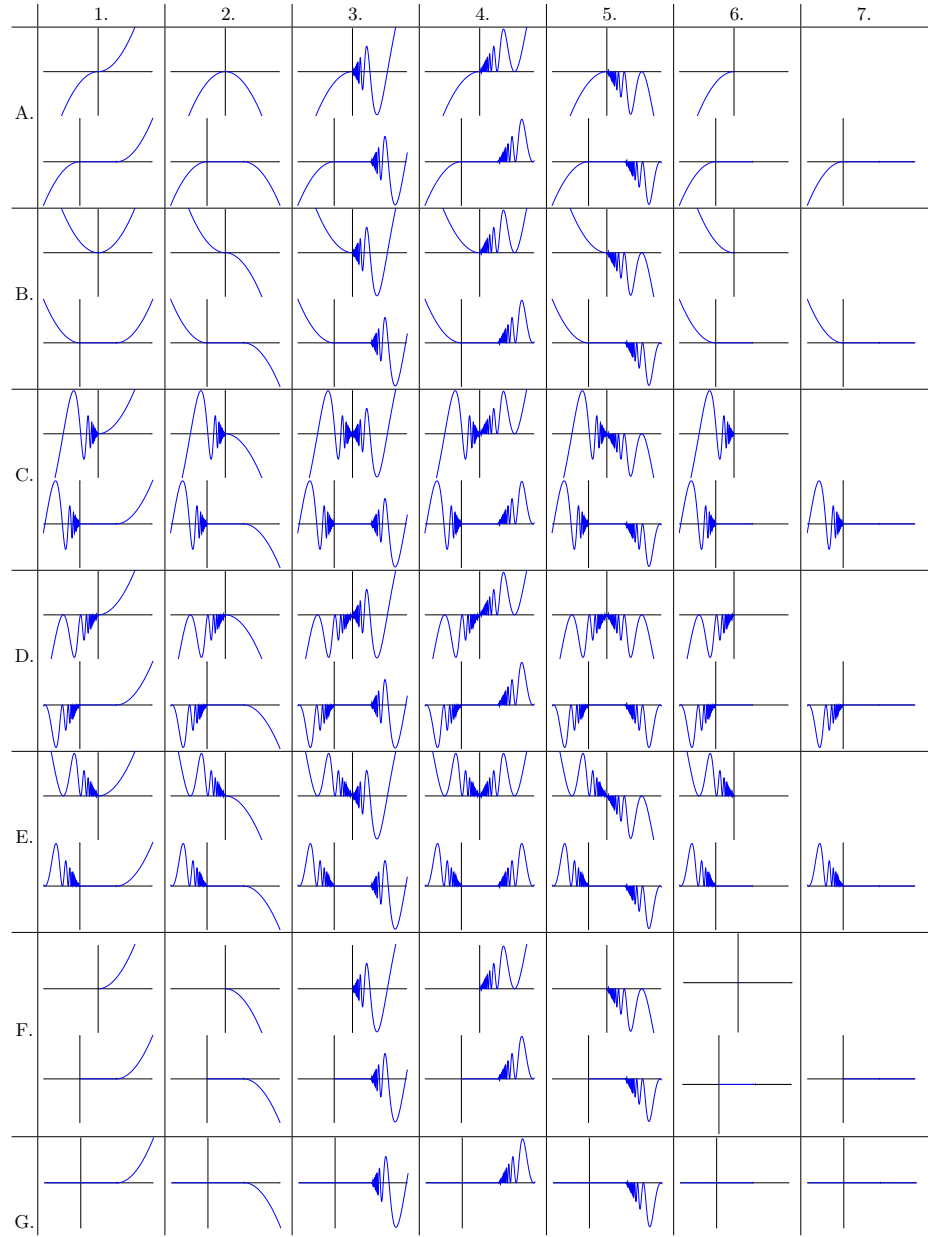





\subsection{Desirable properties}
\begin{enumerate}
    \item If there exist $x,y\in\llparenthesis \ell;u \rrparenthesis$ such that $x<y$, $f(x)<0$ and $f(y)>0$, then there exists a zero-crossing $z\in(x;y)$.
    \item The function $f$ should
be strictly negative somewhere left of, and strictly positive somewhere right of a zero-crossing $z$, formally:\\
$\exists x,y\in\llparenthesis \ell;u \rrparenthesis, x<z<u$, $f(x)<0$ and $f(y)>0$.
    \item A passing case should be a zero-crossing if and only if its corresponding staying case is also a zero-crossing.
\end{enumerate}




Another decision is which point should be the zero-crossing in the staying case where $a<b$. We would typically want the zero-crossing to be $a$ or $b$, but which one is a matter of convention. For this paper, we will pick $b$ (but the subsequent developments can easily be adapted if picking $a$).




In light of these considerations, we define zero-crossings as follows:
\begin{definition}
$z\in\mathbb{R}$ is a zero-crossing for a continuous function $f$ if and only if there exist $a\in\mathbb{R}$ and $b\in\mathbb{R}$ with $a\leq b$ and $z=b$ such that:
\begin{enumerate}
\item $\forall x\in [a, b], f(x)=0$;
\item
$\forall\epsilon>0, \exists x\in[a-\epsilon; a), f(x)<0$;
\item
$\forall\epsilon>0, \exists x\in(b; b+\epsilon], f(x)>0$.
\end{enumerate}
\label{def:second-def-for-zc}
\end{definition}

\section{Syntax}\label{sec:syntax}
We model the hybrid syntax (Fig. \ref{fig:progsyntax}) as a subset of the existing Zélus syntax \cite{benveniste_hybrid_2011}. We extend this core syntax with reset clauses, which are essential for our hybrid examples; although resets are supported by the Zélus implementation, they are omitted from the minimal formal syntax in \cite{benveniste_hybrid_2011} and later introduced as equations in \cite{benveniste2012non}. We syntactically distinguish among expressions that are purely continuous-time, $ce$, purely discrete-time, $de$, and hybrid, $e$. Continuous expressions consist of constants, variables, arithmetic operations, and tuples, and form the fragment that can appear in differential equations and zero-crossing guards. The $\mathsf{last}$ operator represents the left limit of the continuous variable $x$. Discrete expressions are inherited from prior work on a purely discrete subset of Zélus \cite{chen_synchronous_2024}; they include standard synchronous constructs such as local bindings, recursion, function application, and $\mathsf{fby}$, the operator that initializes a stream with one value and specifies its subsequent values with another expression.

The central hybrid construct is \emph{let rec der $\dots$ init $\dots $ reset $\dots$ in $ce$}, which describes a continuous segment governed by an ODE together with event triggered resets that restart the segment when an upward zero-crossing is detected. Tuples provide a uniform way to represent multi-dimensional continuous state, which is needed because many of our examples evolve over several continuous variables simultaneously. Finally, global definitions, $D$, are top-level assignments that form the initial environment of a program and may be used to define constant variables or discrete functions.

\begin{figure}[!t]
    \centering
    \begin{align*}
        ce ::=~& c ~|~ x ~|~\mathsf{last}\ x~|~ ce_1\,\oplus\, ce_2 ~|~ (ce_1,\dots,ce_m)
      & \textnormal{Continuous Expressions}\\
        de ::= &~c~|~x~|~\mathsf{let}_h~(x:\tau)=de_1~\mathsf{in}~de_2 & \textnormal{Discrete Expressions}\\
        &~|~\mathsf{let~rec}_h~(x:\tau)=de_1~\mathsf{in}~de_2 ~|~ f~x~|~de_1~\mathsf{fby}~de_2\\
        &~|~\mathsf{delay}(de)~|~\mathsf{if}~x~\mathsf{then}~de_1~\mathsf{else}~de_2~|~(de_1,...,de_m)\\
        e::=~&\mathsf{let~rec~der}~\mathbf{x}^m=ce_f \\ &\mathsf{init}~de_0~(\mathsf{reset}~\mathsf{up}(ce_i)\to de_{r,i})_{i\in[1,n]}~\mathsf{in}~ce &\textnormal{Hybrid Expressions}\\
        D ::= &~\mathsf{let}~x:\tau=ce ~|~\mathsf{let}~x:\tau=de~& \textnormal{Global Definitions}\\
        &~|~\mathsf{let}~f~(x:\tau_1):\tau_2=de~
    \end{align*}
    \caption{Program Syntax}
    \label{fig:progsyntax}
\end{figure}


\subsection{Specification Syntax}
We derive our specification syntax (Fig. \ref{fig:refsyntax}) from prior work on refinement types for synchronous programming \cite{chen_synchronous_2024}. A refinement type has the form $\{v:b \mid \varphi\}$, where $b$ is a base type and $\varphi$ constrains the values of $v$. In our hybrid setting, however, $\varphi$ is restricted to an invariant trace predicate of the form $\Box p$, where $p$ is a state predicate over continuous expressions. Intuitively, this means that refinements specify safety properties that must hold throughout the execution of a hybrid program, rather than only at a single logical instant.

A key difference is the omission of ``point-wise'' temporal predicates such as $\bigcirc$ and $\mathsf{hd}()$, which apply instantaneously to a specific point in time and are less useful for continuous programs. Consequently, all type refinements become invariant specifications. The syntax of discrete-time types is unchanged apart from requiring invariant types. Although continuous programs are represented in Zélus using floating-point numbers, we make the simplifying assumption that their representation is precise enough to be treated as reals and defer to the existing literature for formally treating floating-point errors. \jiawei{Cite floating point error papers}

\jiawei{The hybrid type formalism only deals with invariant properties; ``point-wise" state predicates don't really make sense}
\begin{figure}[!t]
    \centering
    \begin{align*}
        \tau ::=&~b~\mid~\{v:b\mid\varphi\}~&\textnormal{(Types)}\\
        b ::=&~\kw{float}\mid~b_1\times~b_2~\times~...~\times~b_n~&\textnormal{(Base Types)}\\
        p,q ::=&~\kw{true}~\mid~\kw{false}~\mid~x~\mid~ce_1=ce_2~\mid~ce_1>ce_2& \textnormal{(State Predicates)}\\
        &~\mid~p\land~q~\mid~\neg p\\
        \varphi ::=&~\Box~p& \textnormal{(Trace Predicates)}\\
    \end{align*}
    \caption{Hybrid Type Syntax}
    \label{fig:refsyntax}
\end{figure}

\section{Motivating Examples}\label{sec:examples}
We begin by providing some motivating examples of systems exhibiting hybrid behavior, and demonstrate how they can be verified using our formalism  (Section~\ref{sec:proofs}). Additional examples can be found in the extended version of the paper.
\subsection{Water Tank}
A hybrid system often has multiple reset conditions and these resets do not always have direct influence over the variable of interest. The following program demonstrates a controller which is responsible for filling or draining a tank of water to maintain its level between two setpoints. In this case, the resets do not directly set the water level but instead influence its rate of change. For simplicity, we assume the flow rate can react instantaneously to controller inputs.
\begin{lstlisting}
let rec der (level, flow) =(flow init 5., 0. init 5.) reset
    | up(level -. 9.) -> (last level, -5.)
    | up(-.level +. 1.) -> (last level, 5.)
in (level, flow)
\end{lstlisting}
Ensure: $1 \leq level \leq 9$ at all times



\subsection{Event-Triggered Automatic Braking}
This system features a vehicle equipped with an automatic braking controller modeled after  \cite{butler_adaptive_2011}. The vehicle begins behind a stationary obstacle with an initial velocity and acceleration. The vehicle's starting position is assumed to be far enough behind the obstacle that a collision-free trajectory is possible given the vehicle's maximum braking deceleration. The controller is triggered to react when the position at which the vehicle will stop given maximum braking force begins to surpass the position of the obstacle (plus a safety margin), at which point the controller applies full braking deceleration. The vehicle stops decelerating when its velocity reaches zero. The property to verify is that the vehicle's position never exceeds that of the obstacle, hence remaining collision-free.
\begin{lstlisting}
let x_obs = 5.; brake = -0.2
let rec der (x,v,a)  = (v init 0., a init 1., 0. init 1.) reset 
  | up(x -. (v *. v /. (2. *. brake)) +. 0.1 -. x_obs) -> 
      (last x, last v, brake)
  | up(-. v ) -> (last x, last v, 0.)
in (x, v, a)
\end{lstlisting}
Ensure: $x<x_{obs}$

\section{Hybrid Stream Semantics}\label{sec:semantics}
We extend the semantics of the existing discrete-time synchronous language MARVeLus, \cite{chen_synchronous_2024} which expresses a single step of a stream program's execution with the judgment $S;\sigma \vdash e\stackrel{v}{\hookrightarrow}e'$. A program, represented by $e$, is executed in an environment $S;\sigma$ of functions and terms respectively. The use of separate environments is consistent with previous formalisms \cite{chen_synchronous_2024,benveniste_hybrid_2011}. This results in a value $v$, which is the output of the expression for a given instant in time, as well as a rewritten expression $e'$ to be executed in the next instant. Intuitively, the value $v$ represents the concrete value represented by the program $e$, while $e'$ represents the program's future behavior.

The semantics of MARVeLus are equipped to model a strictly discrete-time subset of Zélus. Consequently, the goal of our formalism is to not only extend the MARVeLus semantics with continuous dynamics, but to also ensure our extension is compatible with the original discrete-time semantics so that the combined semantics can express hybrid behaviors. 

Mixing of continuous and discrete time components is particularly challenging. Fortunately, Zélus addresses this with its handling of physical and logical time, allowing the two to synchronize exclusively at the occurrence of upward zero-crossings \cite{benveniste_hybrid_2011}. When a Zélus program executes, either the continuous or the discrete components are executed, but never both simultaneously. The discrete program executes an entire synchronous iteration in zero physical time, at which point the system integrates the continuous-time dynamics until a zero-crossing detection is triggered. If or when this occurs, the discrete program resumes with logical time advanced by one instant and variable values possibly updated by the continuous dynamics. This cycle repeats indefinitely, resulting in program execution that alternates between continuous- and discrete-time phases.
We note that continuous (hybrid) programs may contain discrete subterms, but not the other way around. This is consistent with the design of Zélus. Therefore, a hybrid program has all its continuous components at the top-level and never inside any discrete components. 


We begin by defining the semantics of continuous expressions and predicates, following a similar convention to dL \cite{platzer2018lfcps}. Given a finite set of variables $Var$, a continuous expression $ce$ of base type $b$ is evaluated in a state $\mathscr{S}:Var\rightharpoonup \mathbb{R}$ as $\llbracket ce \rrbracket:\mathscr{S}\to b$, defined as:
\begin{definition}
For a state $\omega:\mathscr{S}$:
\begin{align*}
    \llbracket x \rrbracket \omega&= \omega(x) & \text{if $x$ is a variable}\\
    \llbracket c \rrbracket\omega &= c & \text{if $c$ is a constant}\\
    \llbracket ce_1~\spadesuit ~ce_2\rrbracket\omega &= \llbracket ce_1\rrbracket\omega ~\spadesuit~ \llbracket ce_2 \rrbracket \omega  & \text{where $\spadesuit$ is a binary arithmetic operator}\\
    \llbracket p~\clubsuit~q\rrbracket \omega &= \llbracket p \rrbracket\omega ~\clubsuit~ \llbracket q\rrbracket\omega & \text{where $\clubsuit$ is a logical connective}\\
    \llbracket \neg p \rrbracket\omega &= \neg \llbracket p \rrbracket\omega\\
    \llbracket (ce_1,...,ce_n)\rrbracket \omega &= (\llbracket ce_1\rrbracket\omega,...\llbracket ce_n\rrbracket\omega)
\end{align*}
\end{definition}

In the case of vector expressions, such as $\mathbf{x}$ defined as $(x_1,...x_n)$, we implicitly apply necessary projection and concatenation operations but otherwise treat them as a single expression. Note that our convention is reversed from that of \cite{platzer2018lfcps} to the expression-then-environment notation, to be more consistent with other operational semantics.

Syntactically determining the solution for a dynamical system is beyond the scope of the present work; we assume the existence of solutions of the dynamical systems. Consequently, we define an abstract solver ``flow'' and its accompanying semantic rule written in terms of such a solution to obtain the continuous evolution of the system until the next zero-crossing:
\[
\inferrule
{
  \forall i.\ \forall t' \in [0,t_r].\ \frac{d\,\rho(t)(x_i)}{dt}(t') \;=\;\llbracket ce_i\rrbracket_{i\in[1,n]}\rho(t')
  \\
  \rho(0)(\mathbf{x}) = \mathbf{v_0}
  \\
  \forall y, \forall i.\ \forall t' \in [0,t_r]. y\neq x_i \implies \left(\rho(y)=\sigma(y) \land \frac{d\,\rho(t)(y)}{dt}(t')=0\right)
  \\
  t_r = +\infty \lor t_r \text{ is the first zero-crossing of } \{\llbracket u_j\rrbracket\rho(t)~|~j\in[1,m]\} 
}
{
  \sigma\vdash\mathsf{flow}\big(\mathsf{der}\ \mathbf{x} = ce,\ \mathbf{v_0},\ \{u_j\}_{j\in[1,m]}\big)
  \Downarrow
  \big(t \mapsto \rho(t),\ t_r\big)
}
\]
  
Here, a dynamical system composed of one or more differential equations, their initial states, zero-crossing detection functions, and the surrounding global environment has a solution given by the time-varying state $\rho(t)$ defined from $t=0$ to $t=t_r\in\mathbb{R}_+^*\cup\{+\infty\}$ (strictly positive or infinity), equivalently $t_r\in(0,+\infty]$ at which point $t_r$ is the earliest time satisfying the chosen zero-crossing predicate for $u$. Consequently, $\rho(t_r)$ represents the system state at the moment a zero crossing occurs (if it reaches the zero-crossing). Continuous global definitions are also imported into $\rho$ from $\sigma$; we follow the convention of dL in requiring that they are constant \cite{platzer2018lfcps}.

\begin{definition}\label{def:firstzc}
For time $t\in\mathbb{R}_+^*\cup\{+\infty\}$, $t$ is the \emph{first zero-crossing} of a set of $n$ functions $\{f_i:\mathbb{R}_+^*\cup\{+\infty\}\to\mathbb{R}~|~i\in [1,n]\}$ if:
\begin{itemize}
    \item $\exists i\in[1,n]~.~t_r\text{ is a zero-crossing of } f_i(t_r)$, and
    \item $\forall t'<t_r~.~\nexists i\in[1,n]~.~t'\text{ is a zero-crossing of } f_i(t_r)$
\end{itemize}
\end{definition}

Here, in $\rho: [0,t_r+\epsilon]\to \mathscr{S}$, for a sufficient $\epsilon$ to verify that $t_r$ is a zero-crossing according to Definition \ref{def:second-def-for-zc} and $\mathscr{S}:Var \rightharpoonup \mathbb{R}$ which associates a finite set of variables to real values.

We define two semantic rules (S-LETREC-DER-FIN) and (S-LETREC-DER-INF) for the two possibilities of hybrid program execution---finite and infinite continuous segments respectively. The judgment returns a ``trace" of the program's continuous execution and the length for which said segment had evolved. In the finite case, the judgment also provides the output of the program after the reset. (S-LETREC-DER-FIN) first evaluates the expression $de_0$ defining the continuous segment's initial condition using discrete computation. Then the semantic judgment on $\mathsf{flow}$ is invoked to obtain a trace of the program's dynamical behavior until a zero-crossing is encountered. It is necessary to disambiguate zero-crossings in the case where two guards $u_i$ and $u_j$ simultaneously approach zero. The semantics of Zélus deterministically chooses the first zero-crossing by syntactic order. Finally, when a zero-crossing is encountered, the $j$-th expression representing its reset is evaluated with the values of the continuous variables immediately prior to reset populated in its context. (S-LETREC-DER-INF) has a similar structure for its initial and continuous phases.

\begin{figure}[!t]
    \centering
        

        \begin{mathpar}
        \inferrule*[lab=\textsc{(S-Letrec-Der-Fin)}]{\\
          S;\sigma \vdash de_0\stackrel{v_0}{\hookrightarrow} (de_0)'\\
          \mathsf{flow}(\mathsf{der}~x^m=ce_f,v_0, (ce_{u,i})_{i\in[1,n]}) \Downarrow (\rho, t_r)\\\;
          t_r<+\infty\\\;
          j=\mathsf{min}(\{j~|~j\in [1,n] \land t_r \textnormal{ is a zero-crossing of } \llbracket ce_{u,j} \rrbracket\rho(t)\})\\\;
          S;\sigma,[x^m\mapsto \rho(t_r)(x^m)::\mathsf{nil}] \vdash de_{r,j}\stackrel{v_{r,j}}{\hookrightarrow}de_{r,j}'
        }{
          S;\sigma \vdash
          \mathsf{let}\ \mathsf{rec~der}\ \mathbf{x}^m:\tau=ce_f~\mathsf{init}~de_0~\mathsf{reset}~(\mathsf{up}(ce_{u,i})\to de_{r,i})_{i\in [1,n]}~\mathsf{in}\ ce
          \\\;\stackrel{(t\mapsto (\llbracket ce \rrbracket\rho(t)), t_r);\llbracket ce\rrbracket [x^m\mapsto v_{r,j}]}{\hookrightarrow}\\\; 
          \mathsf{let}\ \mathsf{rec~der}\ \mathbf{x}^m:\tau=ce_f~\mathsf{init}~v_{r,j}~\mathsf{reset}~(\mathsf{up}(ce_{u,i})\to de_{r,i})_{i\in [1,n]}[de_{r,j}\mapsto de'_{r,j}]~\mathsf{in}\ ce
        }
        \end{mathpar}
        

        \begin{mathpar}
        \inferrule*[lab=\textsc{(S-Letrec-Der-Inf)}]{\\
          S;\sigma \vdash de_0\stackrel{v_0}{\hookrightarrow} (de_0)'\\
          \mathsf{flow}(\mathsf{der}~x^m=ce_f,v_0, (ce_{u,i})_{i\in[1,n]}) \Downarrow (\rho, +\infty)\\\;
        }{
          S;\sigma \vdash
          \mathsf{let}\ \mathsf{rec~der}\ \mathbf{x}^m:\tau=ce_f~\mathsf{init}~de_0~\mathsf{reset}~(\mathsf{up}(ce_{u,i})\to de_{r,i})_{i\in [1,n]}~\mathsf{in}\ ce
          \\\;\stackrel{(t\mapsto (\llbracket ce \rrbracket\rho(t)), +\infty);\mathsf{nil}}{\hookrightarrow}\\\; 
          [\infty]
        }
        \end{mathpar}
                
    \caption{Hybrid Stream Semantics of the form $S;\sigma \vdash e\stackrel{(v(t),t_r); [v_r|\mathsf{nil}]}{\hookrightarrow}e'$}
    \label{fig:semantics}
\end{figure}
 $[\infty]$ denotes the result of a continuous phase which never resets. This is distinct from the classic notion of ``getting stuck" in operational semantics; the program still evolves in continuous time but requires no intervention from discrete resets to remain safe. Thus no discrete computation is triggered.






Note that a rewrite of a hybrid program can produce one of two outputs, respectively for the finite and infinite cases: 
\begin{itemize}
    \item $(t\mapsto \llbracket ce\rrbracket\rho(t), t_r);\llbracket ce \rrbracket [x^m\mapsto v_{r,j}]:\left(([0,t_r]\to b)\times \mathbb{R}_+^*\cup\{\infty\}\right);b$ for programs where the current continuous segment concludes with a zero-crossing detection; $[x^m\mapsto v_{r,j}]$ represents the program's state after the discrete computation following the zero-crossing detection and $b$ is the base type of $ce$, the expression to be evaluated in the context of the hybrid definition.
    \item $(t\mapsto \llbracket ce\rrbracket\rho(t), t_r);\mathsf{nil}:([0,t_r]\to b)\times \mathbb{R}_+^*\cup\{\infty\}$ for programs with an infinite continuous segment and thus do not produce a reset value
\end{itemize}

\subsection{Predicate Semantics}
Verifying that a hybrid program satisfies a given predicate requires formally establishing the requirements for each phase of execution to guarantee its satisfaction. In the discrete-time case \cite{chen_synchronous_2024}, this requires showing that the program's value computed (hence its ``output'') at each instant satisfies the requirements for the predicate at that instant. The hybrid case introduces some additional complexity. Although the program still only computes outputs on discrete-time events, one must also guarantee that the program's dynamics do not violate the specification in the continuous-time phase between these events. Naturally, this requires knowing that a program's execution satisfies the specification for the entire duration of all continuous-time phases, and that furthermore, any discrete computations instantaneously satisfy the specification as well. Since we are primarily interested in verifying safety conditions, we focus on invariant properties of the form $\square p$.

\begin{definition} 
    A hybrid program $e$ satisfies an invariant property $\square q$ if:
    \begin{itemize}
        \item For any continuous execution of $e$ over time $t\in [0,t_r], t_r\in \mathbb{R}_+^*\cup \{+\infty\}$ producing a trace $\rho(t)$, $\llbracket q\rrbracket\rho(t)$ is true for all $t\in[0,t_r]$
        \item For any discrete computation $e_r\stackrel{v_r}{\hookrightarrow}e_r'$ within $e$ which resets its state producing the post-reset state $[x\mapsto v_r]$, $\llbracket q \rrbracket [x\mapsto v_r]$ is true.
    \end{itemize}
\end{definition}
Note that, in our semantics, $t$ represents the time elapsed since the beginning of the \emph{current} continuous segment, not global time. Consequently, we require an invariant property to hold for any continuous or discrete execution that occurs at \emph{any} time.

A pathological case that arises in hybrid systems is Zeno behavior, in which the number of discrete-time events within a given span of time tends towards infinity, as demonstrated by a bouncing ball with decaying velocity. As a result, there is a finite time beyond which the system cannot progress. Hybrid modelers of all kinds must contend with such behaviors and there are approaches to mitigating issues that arise in simulating them \cite{konecny_enclosing_2016}. We make no claim of addressing liveness in these cases. Instead we note that since our predicate semantics is oriented around the execution trace of the program rather than physical time, invariant properties do not make potentially unsafe claims about the system's safety beyond reachable times if Zeno behavior exists. In all other cases, the time of the execution would indeed extend to infinity.

\section{The Hybrid Type System}\label{sec:types}
The type system of the hybrid extension is constructed with adequate conditions so that a well-typed program is guaranteed to provably satisfy a given safety property, encoded as a refinement type. We take inspiration from the sound inference rules of dL and consequently structure our typing rules similarly. 

A crucial difference that arises when applying these techniques to verifying Zélus programs is that there is no exact equivalent to dL's evolution domain, which immediately interrupts continuous evolution if the system begins to exit said domain, in Zélus. While upward zero-crossing detection can approximate an evolution domain, the safe application of this technique is nontrivial. An evolution domain simply examines the instantaneous value of the system to determine whether to interrupt the continuous segment, so it is easy to guard against, for instance, a value exceeding a given number. However, a zero-crossing detector, per Definition \ref{def:second-def-for-zc}, can only do the same if its guard is first negative for a nonzero period of time. As a result, the properties that a zero-crossing detector can ``guard'' against can only be assumed if the zero-crossing detector can be verified to be ``active" at that moment. We introduce the following operator:
\begin{definition}\label{def:active}
    \begin{align*}
    &Active(v, \{u_1,...,u_n\})\triangleq \\ &\qquad \qquad \qquad \qquad \bigwedge\{\square (u_i\leq 0)~|~ \llbracket\left(u_i<0~\lor~(u_i=0\Rightarrow \dot{u_i}<0)\right) \rrbracket[x\mapsto v]\}.
\end{align*}
\end{definition}
which, for a given value $v$ representing the current state of the system, provides the predicates based on the guards of the zero-crossing detectors that can be assumed to be true until at least the next zero-crossing (or for all time if none occur). This is helpful because not all properties can be proven with a differential invariant alone. Thus, this function helps determine which parts of a property can be assumed true in the presence of zero-crossing detectors that will interrupt the program before they are violated. To prove an invariant property $I$ is true for a continuous segment with an initial value $v_0$, it is necessary to find a $\phi$ that is provable via differential invariants such that $\left(\phi \land Active(v_0, \{u_1,...u_n\})\right) \implies I$.

A zero-crossing guard $u_i$ is ``active" if at any point during the execution of the continuous segment $t\in(0,t_r]$, $\llbracket u_i \rrbracket \rho(t)=0$ implies that $t$ is a zero-crossing of $u_i$. Definition \ref{def:active} considers only zero-crossings that can be determined to be active knowing only their instantaneous position and derivative. This is desirable for the proof because it requires no knowledge of the dynamical system's explicit solution. Requiring that $u_i$ begins negative is obvious; it will require a nonzero amount of time for $u_i$ to evolve towards zero thus satisfying Definition \ref{def:second-def-for-zc}. The second condition, also allowing $u_i$ to be considered active if it begins at zero with a negative derivative, arises because $u_i$ will immediately become negative during the continuous segment, thus also requiring nonzero time to continuously evolve towards zero. Here, $\dot{u_i}$ represents the symbolic time derivative of the expression, as defined in \cite{platzer2018lfcps}. This allows, for instance, verifying that a bouncing ball does not fall through the ground immediately after a bounce, since it would be moving away from the ground.




\begin{figure}[!t]
\centering
\begin{mathpar} \inferrule*[lab=\textsc{(T-DER-FIN)}] { (T1)~\exists \phi_I.\Gamma,x:\{v:b|\phi(v)\},\dot{x}:\{v:b | v=e_f\} \vdash  e_0 : \{v:b|\phi(v)\land \phi_I'(\dot{x})\land 
\\ \qquad \qquad \qquad \qquad \qquad \qquad  ((\phi_I(x)\land Active(v,\{u_1...u_n\}) \implies \phi(x))\} \\ (T2)~\forall i\in [1,n],\exists \phi_{i,I}~.~\Gamma, x:\{v:b|\phi(v)\land u_i(v)=0\},\dot{x}:\{v:b | v=e_f\}\vdash \\ \qquad \qquad \qquad \qquad \qquad \qquad e_{r,i} : \{v:b| \phi(v) \land \phi'_{i,I}(\dot{x})\land 
\\ 
\qquad \qquad \qquad \qquad \qquad \qquad ((\phi_{i,I}
(x) \land Active(v,\{u_1...u_n\}) \implies \phi(x))\} \\ (T3)~\Gamma, x:\{v:b|\phi(v)\}\vdash ce : \tau } { \Gamma \vdash \textsf{let rec der}\ x : \{v:b|\phi(v)\} = e_f\ \textsf{init}\ e_0\ \textsf{reset}\ \substack{ \mid\ \textsf{up}(u_1) \to e_{r,1}\\ \vdots\\ \mid\ \textsf{up}(u_n) \to e_{r,n}} \ \textsf{in}\ ce : \tau } \end{mathpar}
\begin{mathpar}
\inferrule*[left=\textsc{(T-DER-INF)}]
{ 
~
}
{ \Gamma \vdash [\infty]:\tau}
\end{mathpar}

\caption{Stream Typing Rules for Hybrid Expressions}
    \label{fig:typing-rules}
\end{figure}

Safety of a program's execution is proven by induction on segments consisting of continuous integration possibly interrupted by a zero-crossing and a discrete computation which initializes the next segment. This is captured in the typing rule (T-DER-FIN) (Fig. \ref{fig:typing-rules}). For each segment, the initial value must first be verified safe. Then, using differential reasoning based on the zero-crossings that are active at the beginning of the continuous period, the continuous part is proven to also obey the invariant. This involves determining which parts of the invariant are proven with knowledge that the active zero-crossings will interrupt execution before they are violated, and then verifying that the remaining portion is provable via a differential invariant. Although these are expressed as existentials in the typing rule, they can be determined based on user-provided annotations or simple heuristic methods in implementation. Concretely, one portion of the safety predicate may be designated by the user as the part to be established by differential reasoning. In this way, the existential form of the rule reflects proof flexibility for the metatheory, while still admitting a practical implementation. 

Finally, all possible resets of the system are proven safe with respect to their triggering conditions to ensure that they are safe initial conditions for subsequent continuous integration. The infinite case types the residual term resulting from rewriting a hybrid expressions with an infinite continuous segment. Since the previous expression would have been well-typed through an application of (T-DER-FIN) which guarantees safety until the next reset, an absence of resets would vacuously imply safety for an infinite amount of time. Thus, no additional verification is necessary.

\subsection{Type Safety}
We prove type safety roughly following the classical syntax-guided manner of Wright and Felleisen \cite{wright1994syntactic}.
\paragraph{Environment correspondence.}
We write $\Gamma \;\approx\; (S;\sigma)$ for the standard correspondence between a typing environment
$\Gamma$ and a runtime configuration $(S;\sigma)$ (control state $S$ with store/history $\sigma$).
Formally,
\[
\Gamma \approx (S;\sigma)
\quad\triangleq\quad
\forall x \in \kw{dom}(\Gamma).\; x \in \kw{dom}(\sigma)\ \land\ (S;\sigma)\vdash \sigma(x) : \Gamma(x).
\]
That is, every variable declared in $\Gamma$ is bound in $\sigma$, and its runtime value is well-typed
with respect to $\Gamma$.

\paragraph{Discrete preservation}
Our hybrid extension shares the discrete sublanguage and its small-step
semantics with MARVeLus. Therefore, we reuse the discrete type preservation result~\cite{chen_synchronous_2024}, specifically Lemma~4 (Type Preservation), for all purely
discrete evaluation and stepping that occurs inside reset branches and other non-continuous terms. Furthermore, correspondence is shown to be preserved by terms introduced into the discrete environment by continuous components.

\begin{lemma}[First up-crossing implies preservation of the active domain]
\label{lem:first-upcross-active-domain}
Assume
\[
flow(\mathsf{der}\ x=e_f,\ x_0,\ \{u_1,\dots,u_n\}) \Downarrow (\rho,t_r)
\]
with $t_r<+\infty$, and suppose branch $i$ triggers at $t_r$ (i.e., $t_r$ is the first upward
zero-crossing among the monitored signals).
Let $\kw{Active}$ be the predicate from Definition \ref{def:active}. Then:
\[
\forall t\in[0,t_r].\ \kw{Active}(\rho(0)(x),\{u_1,\dots,u_n\})
\qquad\text{and}\qquad
u_i(\rho(t_r)(x))=0
\]
where $i$ is the index of the first zero-crossing per Definition \ref{def:firstzc}, i.e. the evolution-domain facts contributed by $\kw{Active}$ are preserved up to $t_r$.
\end{lemma}
Proof can be found in the extended version of the paper.

\begin{lemma}[Monitored-domain preservation for non-terminating flow]
\label{lem:domain-pres-infinite}
Assume
\[
flow(\mathsf{der}\ x=e_f,\ x_0,\ \{u_1,\dots,u_n\}) \Downarrow (\rho,+\infty).
\]
Let $\kw{Active}$ be the predicate from Definition \ref{def:active}. Then:
\[
\forall t\ge 0.\ \kw{Active}(\rho(0)(x),\{u_1,\dots,u_n\}).
\]
\end{lemma}
Proof can be found in the extended version of the paper.

\begin{lemma}[Continuous Type Preservation]
\label{lem:preservation}
    
    A hybrid program $e$ in an environment $\Gamma$ corresponding with $(S;\sigma)$, which is well-typed ($\Gamma\vdash e:\{v:b~|~\square \phi(v)\}$), 
    and steps ($S;\sigma \vdash e \stackrel{(v(t),t_r);[v_r|\mathsf{nil}]}{\hookrightarrow} e'$), has the following hold:
    \begin{enumerate}[leftmargin=*]
        \item $\square_{[0,t_r]}\phi(v(t))$
        \item $\Gamma \vdash v_{r}:\{v:b|~\square \phi(v)\}$ if $t_r$ is finite
        \item $\Gamma\vdash e':\{v:b~|~\square \phi(v)\}$
    \end{enumerate}
\end{lemma}
We prove that a well-typed continuous equation, with sufficiently safe initial conditions per (T-DER-FIN), remains well-typed after one
continuous segment, and that its generated trajectory satisfies the safety predicate throughout the
corresponding integration interval. For the finite case, we use the operational semantics to obtain a trace
$\rho(t)$ up to the first monitored event time $t_r$; Lemma~\ref{lem:first-upcross-active-domain} shows that the
conditions induced by the monitored guards hold for all
$t\in[0,t_r]$.

We then combine these domain facts with the differential invariant provided by the typing
premises (via Platzer's differential invariant rule \cite{platzer2018lfcps}) to conclude that the continuous segment is safe (1). We then show that, in the case of finite $t_r$, the assumptions for discrete type preservation are met, and that the value obtained by executing the discrete program comprising the reset preserves the invariant (2). Furthermore, we show that this constitutes an initial condition sufficiently safe so that the subsequent continuous segment can then be proven safe (3). Per (S-LETREC-DER-FIN), this value is rewritten as the initial condition for the next iteration, thus enabling an inductive proof of safety.
For the infinite case, the same argument applies, except that the monitored-domain lemma holds
for all time, yielding a direct proof of (1) and an inert $[\infty]$ residual that is typed
directly. The complete proof can be found in the extended version of the paper.




\section{Verifying the Examples}\label{sec:proofs}
We provide brief proof sketches for using our formalism to verify safety of the examples introduced in Section \ref{sec:examples}. Detailed proofs can be found in the extended version of the paper.
\subsection{Water Tank}
This system is proven safe by establishing that the invariant $(1 \le \mathit{level} \le 9)$ holds throughout execution. Per \textsc{(T-DER-FIN)}, the initial condition satisfies the invariant, so we may assume safety at the start of the first segment. During each continuous evolution, safety is maintained because the active zero-crossing guards become critical before either bound is violated, forcing a reset while the invariant still holds. Each reset then re-establishes $(1 \le \mathit{level} \le 9)$  and drives the dynamics back inward, yielding a safe starting state for the next segment. Repeating this argument over successive continuous segments and resets yields, by induction, that the invariant holds globally.

\subsection{Event-Triggered Automatic Braking}
This system is proven safe by first proving that the invariant $\square(x-\frac{v^2}{2b}<x_{obs})$ holds. Per (T-DER-FIN), the initial condition can be proven safe with the trivial differential invariant $\phi_I=\mathsf{true}$ since the zero-crossing guard, which is identical to the invariant, is negative. Since the reset expression sets the acceleration exactly equal to the acceleration required to stop at exactly the position of the obstacle, the zero-crossing guard stays at zero and thus the invariant must shift back to $\phi_I$ to be proven using the differential invariant technique.

More examples and their corresponding proofs can be found in the extended version of the paper.

\section{Related Work}\label{sec:related}

%

\paragraph{Hybrid Semantics: }
Precision is essential in hybrid semantics, and is particularly challenging when discrete program executions must be scheduled along with continuous dynamics, as imprecise semantics leads to unpredictable models of their interactions. This has been addressed in the Zélus language \cite{benveniste_hybrid_2011} which restricts interactions between continuous and discrete executions to mitigate inconsistencies that may arise from solver implementation. Benveniste \emph{et al.} characterize the hybrid semantics of Zélus with a \emph{non-standard} semantics based on infinitesimal time steps and idealized micro-steps~\cite{benveniste2012nonstandard}. In contrast, our work stays within the conventional time representation common to other formalisms of hybrid semantics but makes the semantics of event detection (zero-crossings) explicit so that sound reasoning principles can be stated directly. Differential dynamic logic (dL) provides a
compositional deductive verification framework for hybrid systems, with proof rules tailored to
hybrid programs and differential equations \cite{platzer2008dl}. Operational accounts in Ptolemy II
and related work emphasize executability, determinism, and discrete/continuous interaction (often
via superdense time) \cite{lee2005operational}. Our work is complementary: rather than proposing a
general proof calculus or a particular execution model, we formalize language-level event detection
(zero-crossings) and its interaction with continuous evolution so that metatheory and type-based
verification can rely on a precise account of events.

\paragraph{Verifying Hybrid Systems and CPS: }
Verification of hybrid systems has been extensively studied in control theory, with approaches like barrier certificates \cite{prajna_safety_2004} which can efficiently provide sufficient safety conditions. However, there is also interest in verifying systems described by programming languages as they may offer a more straightforward transfer to real-world implementations. Lingua Franca makes timing and concurrency semantics explicit for CPS using specific models of
computation, enabling deterministic-by-construction designs \cite{lin2023lf}. Recent
work also formalizes subsets of Lingua Franca in rewriting logic to support analysis in Maude
\cite{marin2024lfmaude}. These efforts reinforce the broader trend of semantics-driven verification;
our contribution addresses the analogous need for hybrid programs by establishing zero-crossings and
continuous evolution in a verification-friendly semantics. Similarly, KeYmaera X is an interactive and automated theorem prover for hybrid systems based on dL, with tactic-based proof search and a small trusted core \cite{fulton2015keymaera}. Furthermore, Plaidypvs \cite{white2024temporal} embeds dL into PVS, supporting standard dL-style proofs while additionally enabling reuse of existing PVS theories to reason over broad classes of hybrid programs \cite{white2024temporal}. While our approach in comparison is more limiting in the systems that are verifiable, we accept this compromise in the interest of pursuing automated verification via the type checker and direct execution of verified source code. Another related work is HABS, a hybrid programming language for formal modeling and verification of hybrid systems \cite{kamburjan2022hybrid}, which uses a similar synchronization mechanism. HABS translates programs into differential dynamic logic and discharges the resulting proof obligations in KeYmaera X, whereas our work targets Zélus and performs verification at compile time via refinement typing. There are also specification-based runtime monitoring approaches for cyber-physical systems~\cite{bartocci2018specification}. These methods monitor and predict CPS behaviors during simulation-time or at runtime, while our work instead targets compile-time verification through a refinement type system for statically establishing safety properties.
\paragraph{Verification via Types: }
Liquid types enable automated type-checking by restricting refinements to a
decidable fragment \cite{rondon2008liquid}, and LiquidHaskell shows these ideas scale to realistic
programs without changing the dynamic semantics \cite{vazou2014liquidhaskell}. This motivates using
types as a lightweight automated verification interface. Our setting differs in that correctness depends not
only on value-level invariants but also on the meaning of continuous evolution and event detection,
which necessitates a formal characterization of zero-crossings for type soundness and verification.

\section{Discussion}
The formalization of zero-crossings and hybrid synchronous semantics lays the critical groundwork for a language capable of simultaneous formal verification and execution. Our framework interprets a hybrid program's execution as a stream of alternating continuous and discrete executions. This enables invariants to be proven via induction. Due to the unique semantics induced by the Zélus zero-crossing detector $\mathsf{up}()$, establishing sufficient conditions to successfully prove safety via induction proved to be surprisingly complex. Nevertheless, this method provides an approach to syntactically verify the correctness of a synchronous program with hybrid components using tactics inspired by differential dynamic logic \cite{platzer2018lfcps}. Compared to \emph{dL}, synchronizing discrete events is more challenging due to our handling of zero-crossings. We gain an executable, compositional language-level proof interface but lack an explicit evolution-domain constraint $Q$ as in \emph{dL}. Compared to \emph{hybrid automata}, we gain determinism and code alignment, but must handle low-level event corner cases more explicitly. Our current formalization handles discrete transitions induced by zero-crossing events of continuous expressions, but does not yet account for inherently discrete or exogenous events outside that mechanism.

Although the restriction of our predicates to strictly invariant properties allowed for simpler proofs, it may be useful to expand the proof rules to account for other temporal operators or logics, like Signal Temporal Logic which ranges over time intervals. Furthermore, the formalism would benefit greatly from automation in an SMT-based type system like that of Liquid Haskell \cite{vazou2014liquidhaskell}, or combined with MARVeLus \cite{chen_synchronous_2024} to form a complete, automatically verifiable hybrid subset of Zélus. Nevertheless, in its current state, our formalism demonstrates that it is possible to characterize zero-crossings and synchronous hybrid programs that rely on them to verify safety properties.



\section*{Acknowledgments}
The authors would like to thank Timothy Bourke, Gr\'egoire  Bussone, Charles de Haro, Paul Jeanmaire and Partha Roop for interesting discussions on earlier versions of this paper.
This work was funded in part by National Science Foundation grants CCF-2348706 and CCF-2426474, and the U.S. Federal Aviation Administration through the FAST Grant Program, Grant Number NG-FAS-0016, Funding Opportunity Number 23-FAA-FAST-001 under the supervision of Joshua Glottmann. Any opinions, findings, conclusions or recommendations expressed in this material are those of the authors and do not necessarily reflect the views of the FAA. 
%
%
%
\newpage
\bibliographystyle{splncs04}
\bibliography{refs}
\newpage
\appendix 
\section{Case Distinction}
\label{app:cases}
We have identified the following 7 cases for $x<a$:
\begin{enumerate}[label=\Alph*]
    \item $f$ is strictly negative in a ball left of $a$. Formally:
    $\ell < a$ and $\exists \eta>0, \forall x \in [a-\eta, a), f(x)<0$.\\
    Examples: $(x-a)^{2k+1}$, $-|x-a|^k$, $-\exp(-\frac1{(x-a)^2})$.
    \item $f$ is strictly positive in a ball left of $a$. Formally:
    $\ell < a$ and $\exists \eta>0, \forall x \in [a-\eta, a), f(x)>0$.\\ 
    Examples: $(x-a)^{2k}$, $|x-a|^k$, $\exp(-\frac1 {(x-a)^2})$.
    \item $f$ is both strictly positive and strictly negative in any ball left of $a$. Formally:
    $\ell<a$ and 
    $\forall \eta>0, \exists x,y \in [a-\eta, a)$ such that $f(x)>0$ and $f(y)<0$.\\
    Examples: $(x-a)^k\sin(\frac1{x-a}) $, $ \exp(-\frac1 {(x-a)^2})\sin(\frac1{x-a}) $.
    \item $f$ is negative or zero in a ball left of $a$, and is both negative and zero in any ball left of $a$. Formally:
    $\ell < a$ and $\exists \eta>0, \forall x \in [a-\eta, a), f(x)\leq 0$ and 
    $\forall \eta>0, \exists x,y \in [a-\eta, a)$ such that $f(x)<0$ and $f(y)=0$.\\
    Examples: $-(x-a)^{2k}\sin^2(\frac1{x-a}) $, $ -\exp(-\frac1 {(x-a)^2})\sin^2(\frac1{x-a}) $.
    \item $f$ is positive or zero in a ball left of $a$, and is both positive and zero in any ball left of $a$. Formally:
    $\ell < a$ and $\exists \eta>0, \forall x \in [a-\eta, a), f(x)\geq 0$ and
    $\forall \eta>0, \exists x,y \in [a-\eta, a)$ such that $f(x)>0$ and $f(y)=0$.\\
    Examples: $(x-a)^{2k}\sin^2(\frac1{x-a}) $, $ \exp(-\frac1 {(x-a)^2})\sin^2(\frac1{x-a}) $.
    \item $\ell = a \in\mathbb{R}$. This is equivalent to $\forall x \in [\ell; b], x=0$.
    \item $\ell = a = -\infty$.
\end{enumerate}

The 7 cases for  $x>b$ are symmetric:
\begin{enumerate}
    \item $f$ is strictly positive in a ball right of $b$. Formally:
    $b < u$ and $\exists \eta>0, \forall x \in (b, b+\eta], f(x)>0$.\\
    Examples: $(x-b)^k$, $|x-b|^k$, $\exp(-\frac1{(x-b)^2})$.
    \item $f$ is strictly negative in a ball right of $b$. Formally:
    $b < u$ and $\exists \eta>0, \forall x \in (b, b+\eta], f(x)<0$.\\
    Examples: $-(x-b)^k$, $-|x-b|^k$, $-\exp(-\frac1 {(x-b)^2})$.
    \item $f$ is both strictly positive and strictly negative in any ball right of $b$. Formally:
    $b<u$ and $\forall \eta>0, \exists x,y \in (b, b+\eta]$ such that $f(x)>0$ and $f(y)<0$.\\
    Example: $(x-b)^k\sin(\frac1{x-b})$.
    \item 
    $f$ is positive or zero in a ball right of $b$, and is both positive and zero in any ball right of $b$. Formally:
    $b < u$ and $\exists \eta>0, \forall x \in (b, b+\eta], f(x)\geq 0$ and 
    $\forall \eta>0, \exists x,y \in (b, b+\eta]$ such that $f(x)>0$ and $f(y)=0$.\\
    Example: $(x-b)^{2k}\sin^2(\frac1{x-b})$. 
    \item 
    $f$ is negative or zero in a ball right of $b$, and is both negative and zero in any ball right of $b$. Formally:
    $b < u$ and $\exists \eta>0, \forall x \in (b, b+\eta], f(x)\leq 0$ and 
    $\forall \eta>0, \exists x,y \in (b, b+\eta]$ such that $f(x)<0$ and $f(y)=0$.\\
    Example: $ -(x-b)^{2k}\sin^2(\frac1{x-b}) $. 
    \item $b = u \in\mathbb{R}$. This is equivalent to $\forall x \in [a; u], x=0$.
    \item $b = u = +\infty$.
\end{enumerate}

The 7 cases for $x<a$ and the 7 cases $x>b$ give rise to 49 cases, most of them subdivided into a passing subcase and a staying subcase, bringing the total to 85. 
All those cases are summarized in Figure~\ref{fig:zero-crossing-cases}. Line G and column 7 do not have passing subcases, because it does not make sense to have $a=b$ with $a=-\infty$ or $b=+\infty$.

We now turn our interest to understanding the different possibilities for $f$ when $x<a$ and $x>b$. We identify 7 different cases when $x<a$ (numbered A to G), as well as 7 different cases when $x>b$ (numbered 1 to 7). All cases are explained using examples in Appendix~\ref{app:cases} and they are also summarized in Figure~\ref{fig:zero-crossing-cases}. In examples of this section, when a function is undefined at a point, by convention we extend it (by continuity) by zero at this point.



\subsection{Possible Formal Definitions of a Zero-Crossing}
Now that we have identified all possible cases, 
We are now ready to decide which cases should be zero-crossings, and which cases should not. 
But first we identify desirable properties of the definition.



\subsubsection{Formal definitions of Zero-Crossings}

To respect Desirable Property 2 in Section \ref{sec:zc}, we eliminate columns 2, 5, 6 and 7; as well as lines B, E, F and G. We believe that A1, A4, D1 and D4 should clearly be zero-crossings. We believe that column 3 and line C are more controversial as to whether they should have a zero-crossing, or not, or throw an error.

Given those ideas, if we limit the zero-crossings to cases A1, A4, D1 and D4, we can come up with the following definition of a zero-crossing:
\begin{definition}
$z\in\mathbb{R}$ is a zero-crossing for function $f$ if and only if there exist $a\in\mathbb{R}$ and $b\in\mathbb{R}$ with $a\leq b$ and $z=b$ such that:
\begin{enumerate}
\item $\forall x\in [a, b], f(x)=0$;
\item
$\forall\epsilon>0, \exists x\in[a-\epsilon; a), f(x)<0$;
\item
$\forall\epsilon>0, \exists x\in(b; b+\epsilon], f(x)>0$;
\item 
$\exists\eta>0, \forall x\in[a-\eta; a), f(x)\leq 0$;
\item
$\exists\eta>0, \forall x\in(b; b+\eta], f(x)\geq 0$.
\end{enumerate}
\end{definition}

Now, if we consider that beyond cases A1, A4, D1 and D4, cases A3, C1, C3, C4 and D3 should also be considered zero-crossings, we can come up with a looser definition of a zero-crossing, only including conditions (i), (ii) and (iii), which results in the definition that appears in the main paper:
\begin{definition}
$z\in\mathbb{R}$ is a zero-crossing for function $f$ if and only if there exist $a\in\mathbb{R}$ and $b\in\mathbb{R}$ with $a\leq b$ and $z=b$ such that:
\begin{enumerate}
\item $\forall x\in [a, b], f(x)=0$;
\item
$\forall\epsilon>0, \exists x\in[a-\epsilon; a), f(x)<0$;
\item
$\forall\epsilon>0, \exists x\in(b; b+\epsilon], f(x)>0$.
\end{enumerate}
\label{def:second-def-for-zc2}
\end{definition}

Both definitions respect Desirable Properties 1, 2 and 3. 
We can also imagine hybrid definitions that include conditions (i), (ii), (iii) and (iv); or (i), (ii), (iii) and (v) of Definition 1. In the rest of this paper we will consider Definition 2 in our semantics for zero-crossings.  

Another decision is which point should be the zero-crossing in the staying case where $a<b$. We would typically want the zero-crossing to be $a$ or $b$, but which one is almost purely a matter of convention. For this paper, we will pick $b$ (but the subsequent developments can easily be adapted if picking $a$).


\section{Proof of Lemma \ref{lem:first-upcross-active-domain}}\label{app:lem1}



\begin{proof}
Define $g_j(t) \triangleq u_j(\rho(t)(x))$. Each $g_j$ is continuous in $t$
(because $\rho$ is continuous and $u_j$ is a continuous expression).
    Fix any $j$ where $j$ is the index of an active guard. We prove $\forall t\in[0,t_r],\ g_j(t)\le 0$ by contradiction. Assume there exists $t_1 \in [0,t_r]$ such that $g_j(t_1) > 0$.

    First, we show that $g_j$ is strictly negative at some time strictly after $0$:
    \begin{itemize}
        \item If $g_j(0)<0$, then $g_j$ is negative at time $0$.
        \item If $g_j(0)=0$ and $\dot g_j(0)<0$, then by differentiability at $0$ there exists $\delta>0$ such that for all $t\in(0,\delta]$, $g_j(t)<0$.
    \end{itemize}

    In either case, there exists some $t_- \in [0,t_1)$ with $g_j(t_-)<0$.

    Assume towards a contradiction that $g_j$ becomes positive at or before $t_r$. Then there exists $t_1 \le t_r$ such that $g_j(t_1)>0$. Let
    \[
    \tau \triangleq \inf\{t\in[0,t_1]\mid g_j(t)>0\}.
    \]
    By continuity, $g_j(\tau)=0$, and by definition of $\tau$ there are arbitrarily close times after $\tau$ where $g_j$ is positive. Also, since $g_j(t_-)<0$ for some $t_-<\tau$, there are negative values somewhere before $\tau$.

    Now consider the maximal interval on which $g_j$ is $0$ ending at $\tau$: choose any $a\le \tau$ such that $g_j(t)=0$ for all $t\in[a,\tau]$ and $a$ is minimal with this property. Then:
    \begin{itemize}
        \item $\forall t\in[a,\tau]$, $g_j(t)=0$ by construction;
        \item for every $\varepsilon>0$, there exists $t\in[a-\varepsilon,a)$ with $g_j(t)<0$
(by minimality of $a$);
        \item for every $\varepsilon>0$, there exists $t\in(\tau,\tau+\varepsilon]$ with $g_j(t)>0$
(by definition of $\tau$).
    \end{itemize}
    These are exactly the conditions of Def.~2, so $\tau$ is an upward zero-crossing
time of $g_j$. Finally, because $g_j(t_1)>0$ occurs at some $t_1\le t_r$, we have $\tau < t_r$,
so we have found an upward zero-crossing strictly before $t_r$. This contradicts
that $t_r$ is the \emph{first} upward zero-crossing among the monitored signals
(Def.~4). Hence $g_j(t)\le 0$ for all $t\in[0,t_r]$.

Since $j$ was arbitrary among the active guards, all conjuncts produced by
$\mathrm{Active}(\rho(0)(x),\{u_1,\dots,u_n\})$ are preserved up to $t_r$.

Moreover, since $t_r$ is an upward zero-crossing time of $g_i$ (Def.~2), in particular
$g_i(t_r)=0$, i.e.\ $u_i(\rho(t_r)(x))=0$.
\end{proof}
\paragraph{Remark (Non-active guards may trigger first).}
Note that Lemma~\ref{lem:first-upcross-active-domain} does \emph{not} require
that the first triggered event comes from an active guard. Even if some
initially non-active guard triggers the first up-crossing at time $t_r$, the
proof above shows that \emph{every} guard that \emph{was} active at time $0$
still cannot become positive before $t_r$ (otherwise it would generate an
earlier up-crossing), matching the stated intuition.
\section{Proof of Lemma \ref{lem:domain-pres-infinite}}\label{app:lem2}

\begin{proof}
    Fix an active guard index $j$ and let $g_j(t)=u_j(\rho(t)(x))$ as in~\ref{app:lem1}.
As in Lemma~\ref{lem:first-upcross-active-domain}, activity at time $0$ implies
$g_j$ is negative at some time $t_- \ge 0$.

Assume for contradiction that $g_j(t_1)>0$ for some $t_1\ge 0$.
Define $\tau = \inf\{t\in[0,t_1]\mid g_j(t)>0\}$.
The same continuity argument as before shows that $\tau$ satisfies Def.~2,
hence $\tau$ is an upward zero-crossing time of $g_j$.

But this contradicts the assumption that the flow segment has no monitored
upward zero-crossing at any finite time (the $+\infty$ outcome of $\mathrm{flow}$).
Therefore $g_j(t)\le 0$ for all $t\ge 0$.
\end{proof}
\section{Proof of Continuous Type Preservation}\label{app:preservation}
\begin{proof}
By structural induction on the operational semantics derivation of $S;\sigma \vdash e \stackrel{(\rho,t_r)}{\hookrightarrow} e'$. 

\paragraph{Case \textsc{S-LETREC-DER-FIN}.}
Assume $\Gamma \vdash e : \{v:b \mid \square~ \phi(v)\}$ and
      $S;\sigma \vdash e \stackrel{(v(t),t_r);v_r}{\hookrightarrow} e'$ is derived by \textsc{S-LETREC-DER-FIN},
      where \begin{enumerate}[leftmargin=*]
          \item $flow(\textsf{der }x=e_f,\;v_0,\;\{u_1(x),\dots,u_n(x)\}) \Downarrow (\rho,t_r)$ with $t_r<+\infty$,
        \item there exists an index $j$ whose guard triggers at $t_r$ (first upward zero-crossing),
        \item and the reset branch steps in the discrete fragment under the post-flow store $S;\sigma,\,[x\mapsto\rho(t_r)(x)::\mathsf{nil}]$, producing a rewritten reset expression $de'_{r,j}$,
and the equation residual $e'$ is obtained by replacing the initial value with the
post-reset state and the reset branch to $de'_{r,j}$, where $v_{r,j}$ is emitted.
      \end{enumerate}

\paragraph{(1) Safety of the finite continuous trajectory.}\textsc{T-DER-FIN} must be the last (non-subtyping) rule used. By inversion, we obtain its premises.
\begin{description}
\item[(T1)] There exists an auxiliary invariant $\varphi_I$ such that typing $e_0$ establishes:
(i) the safety predicate $\varphi$ at time $0$,
(ii) the differential-invariant side condition needed by Platzer's rule (dI) for $\varphi_I$,
and (iii) a bridge implication from $\varphi_I$ together with the $\text{Active}(v,u_1(x) ... u_n(x))$ to $\varphi$.
\end{description}
From the flow premise and Lemma~\ref{lem:first-upcross-active-domain} due to a finite $t_r$, we have:
\[
\forall t\in[0,t_r].\ \text{Active}(\rho(t)(x))
\qquad\text{and}\qquad
u_i(\rho(t_r)(x))=0.
\]
Now apply the soundness of Platzer's differential invariant rule (dI) \cite{platzer_complete_2017}:
premise (T1) provides the initial establishment of $\varphi_I$ together with its differential
side condition, and Lemma~\ref{lem:first-upcross-active-domain} supplies $\text{Active}$ throughout $[0,t_r]$.
Therefore $\varphi_I$ holds at all times in $[0,t_r]$. Finally by assumption from (T1), 
\[
\square_{[0, t_r]}\phi(v(t))
\]

We then show $\Gamma_c \vdash e' : \{v:b\mid \Box \phi(v)\}$ by re-applying \textsc{T-Der-Fin}
to the rewritten program.

\smallskip
\noindent\emph{(Body premise.)}
The typing premise for the body (call it (T3)) is unchanged: the binder for $x$ in the residual
remains $x:\{v:b\mid \phi(v)\}$, hence the same derivation establishes the body typing.

\smallskip
\noindent\emph{(Reset premise and updated initializer.)}
Since branch $j$ triggers at time $t_r$, then we have $u_j(x_r)=0$ from Lemma~\ref{lem:first-upcross-active-domain}. Together with (1) we have $x_r\models \phi$.
Therefore, the runtime store used to evaluate the reset expression,
$\sigma[x\mapsto x_r]$, satisfies the assumptions of the typing premise (T2) for branch $i$,
which types $de_{r,j}$ under
\[
x:\{v:b\mid \phi(v)\wedge u_i(v)=0\}.
\]
Since the predicate syntax of the hybrid type system is a subset of those in the purely discrete language, correspondence of the typing and term environments in the hybrid context implies correspondence in the discrete context as well.
By the assumed discrete preservation for the non-continuous fragment,
from $S;\sigma,[x^m\mapsto \rho(t_r)(x^m)::\mathsf{nil}] \vdash de_{r,j}\stackrel{v_{r,j}}{\hookrightarrow}de_{r,j}'$ we obtain that the result $v_{r,j}$
satisfies the refinement postcondition stated by (T2). This also uses the fact that $\mathsf{tl}(\square p)\equiv \square p$. In particular, since (T2) bundles
the conjunct $\phi(v)$ in the type of $de_{r,i}$, we conclude
\[
v_{r,j} \models \phi.
\]
All other reset expressions retain their previous typing judgments due to them not stepping. 
Finally, \textsc{S-Letrec-Der-Fin} updates the initializer of the residual program to be this
post-reset value $v_{r,j}$. Hence, the initializer premise required to re-apply \textsc{T-Der-Fin}
holds for $ce'$.

\smallskip
\noindent
Combining the unchanged body premise with the updated initializer and the preserved reset-branch
typing, all premises of \textsc{T-Der-Fin} hold for $ce'$, and therefore
\[
\Gamma_c \vdash v_{r,j} : \{v:b \mid \Box \phi(v)\}  \qquad and \qquad \Gamma_c \vdash ce' : \{v:b \mid \Box \phi(v)\}.  
\]

This completes the S-LETREC-DER-FIN case.

\paragraph{Case \textsc{S-LETREC-DER-INF}.}
The last operational rule is \textsc{S-Der-Inf}. Thus
\[
flow(\mathsf{der}\ x=e_f,\ x_0,\ \{u_1,\dots,u_n\}) \Downarrow (\rho,+\infty),
\]
and the residual is the distinguished infinite marker $[\infty]$.

We must show:
\begin{enumerate}
\item $\llbracket v\rrbracket \rho \vDash \Box_{[0,+\infty)} \phi(v)$, i.e. $\forall t\ge 0.\ \rho(t)(x)\models \phi$;
\item $\Gamma_c \vdash [\infty] : \{v:b \mid \Box \phi(v)\}$.
\end{enumerate}

\paragraph{(1) Safety of the infinite continuous trajectory.}
By inversion, the last typing rule is again \textsc{T-Der-Fin} (or the corresponding continuous typing
rule for derivatives), yielding premise (T1) with an auxiliary invariant $\varphi_I$ and a bridge.
From Lemma~\ref{lem:domain-pres-infinite} we have domain preservation for all time:
\[
\forall t\ge 0.\ \text{Active}(\rho(t)(x)).
\]
By soundness of (dI), the differential side condition from (T1) implies that $\varphi_I$ is preserved
along $\rho$ for all $t\ge 0$ (since $\text{Active}$ holds for all $t\ge 0$). Applying the bridge pointwise
yields:
\[
\forall t\ge 0.\ \rho(t)(x)\models \varphi,
\]
i.e.\ $\rho \models \Box_{[0,+\infty)}\varphi$.

\paragraph{(2) Residual typing.}
The residual is $[\infty]$ (or equivalent). By the typing rule for the inert infinite residual,
we derive directly:
\[
\Gamma_c \vdash [\infty] : \{v:b \mid \Box\varphi(v)\}.
\]

\noindent
This concludes both cases.
\end{proof}
\section{Detailed Proofs of Hybrid Examples} 
\label{app:exampleproofs}
We provide detailed proofs of the two examples presented in the text, and introduce two additional examples.
\subsection{Proof of the Water Tank Example}
\begin{align*}
\text{T3}\equiv &~\Gamma,\;
  \mathit{(level, flow)} : \{(l,f):fl\times fl \mid \Box(1 \le l \le 9)\}
  \;\vdash\;\\&
  level : \{l : fl \mid \Box(1 \le l \le 9)\}
\end{align*}

\begin{align*}
\textsc{T1}\equiv&~\Gamma,\;
  \mathit{level} : \{v:fl \mid \square(1 \le v \le 9)\},\;
  \dot{\mathit{level}} : \{v:fl \mid v = \square\mathit{flow}\}
  \;\vdash\;\\&
  (5, 5) :
  \Bigl\{
    (l,f):fl\times fl \mid
      \square(1 \le l \le 9)
      \land \\ &\dot{True}
      \land
      \bigl(
        (True \land \square(1\le l \le 9))
        \Rightarrow \square
        (1 \le level \le 9)
      \bigr)
  \Bigr\}
\end{align*}

\begin{align*}
\textsc{T2$_1$}\equiv
  &~\Gamma,\;
  \mathit{level} :
    \{\square (v \mid 1 \le v \le 9) \land v = 9\},\;
  \dot{\mathit{level}} : \{v \mid v = \square\mathit{flow}\}
  \;\vdash\;\\
  &(9,-5) :
  \Bigl\{
    (l,f):fl\times fl \mid
      \square(1 \le l \le 9)
      \land \square\dot{(level \le 9)}
      \land \\&
      \bigl(
        (\square(level \le 9) \land \square(1 \le l \le 9))
        \Rightarrow
        \square(1 \le level \le 9)
      \bigr)
  \Bigr\}
\end{align*}

\begin{align*}\textsc{T2$_2$}\equiv~
  &\Gamma,\;
  \mathit{level} :
    \{v \mid \square(1 \le v \le 9) \land v = 1\},\;
  \dot{\mathit{level}} : \{v \mid v = \square\mathit{flow}\}
  \;\vdash\;\\&
   (1,5) :
  \Bigl\{
    (l,f):fl\times fl \mid
      \square(1 \le l \le 9)
      \land \square\dot{(1 \le level)}
      \land \\&
      \bigl(
        (\square(1 \le level) \land \square(1 \le l \le 9))
        \Rightarrow
        \square(1 \le level \le 9)
      \bigr)
  \Bigr\}
\end{align*}

\begin{mathpar}
\inferrule*[lab=\textsc{T-DER-FIN}]
{
  \textsc{T1}
  \\
  \textsc{T2$_1$}
  \\
  \textsc{T2$_2$}
  \\
  \textsc{T3}
}
{
  \Gamma \vdash
  \text{let rec }
    \text{der (level, flow)} : \{(l,f) : fl \times fl \mid \Box(1 \le l \le 9)\} \\ = \text{(flow, 0) init (5
, 5)}~
    \text{reset}
    \substack{\mid\ \text{up(level-9)} \to (\text{last level},-5)\\\mid\ \text{up(-level+1)} \to (\text{last level}, 5)} \\
  \;\text{in}\; \text{level}
  : \{l : fl \mid \Box(1 \le l \le 9)\}
}
\end{mathpar}

\subsection{Proof of the Event-Triggered Automatic Braking Example}
Let $prog=$
\begin{lstlisting}
let x_obs = 5.; brake = -0.2
let rec der (x,v,a)  =  
    (v init 0.,
     a init 1.
     0. init 1.) reset 
  | up(x -. (v *. v /. (2. *. brake)) +. 0.1 -. x_obs) -> 
      (last x, last v, brake)
  | up(-. v ) -> (last x, last v, 0.)
in (x, v, a)
\end{lstlisting}

The program initially starts with both guards $x-\frac{v^2}{2\times brake}+0.1-x_{obs}$ and $-v$ negative; thus the differential invariant is not necessary to prove safety since these alone imply the main invariant.

Upon a zero-crossing on $x-\frac{v^2}{2\times brake}+0.1-x_{obs}$, the acceleration is set to $brake$. However, since $x$ and $v$ retain their previous values and the derivative of this guard is zero, this guard is no longer active. Therefore, use the differential invariant to prove safety. This is done by observing that, since $brake$ is negative, $v-\frac{va}{brake}\leq 0$ only if $a$ is at least as negative as $brake$. This is obviously true so the differential invariant is provable.

Upon a zero-crossing on $-v$, the acceleration is set to $0$. As before, neither this guard nor the other can be proven to be active after reset, so the differential invariant is used again. Since $v=0$, $v-\frac{va}{brake}\leq 0$ trivially and furthermore it is provable to stay at zero because acceleration is also zero, thus guaranteeing safety.

Applying (T-DER-FIN):
\begin{mathpar}
{T3}\equiv

{
  \Gamma,\;
  \mathit{(x, v, a)} : \{(v_x, v_v, v_a):fl\times fl\times fl \mid \Box(v_x - \frac{v_v^2}{2\times brake}<x_{obs})\}}\\
  {\;\vdash\;
  (x, v ,a) : \{(v_x, v_v, v_a) : fl\times fl \times fl \mid\Box(v_x - \frac{v_v^2}{2\times brake}<x_{obs})\}
}
\end{mathpar}

For $e_0=(0, 1, 1)$, select $\phi_I=True$
\begin{mathpar}
\inferrule*[lab=\textsc{T1$\equiv$}]
{
  \Gamma,\;
  \mathit{(x,v,a)} : \{(v_x, v_v, v_a):fl\times fl \times fl \mid \Box(v_x-\frac{v_v^2}{2\times brake}<x_{obs})\},\;\\
  \mathit{(\dot{x}, \dot{v}, \dot{a})} : \{(w_x, w_v, w_a):fl\times fl\times fl \mid \Box(w_x, w_v, w_a)=(v, a, 0)\}
  \;\vdash\;\\
  (0, 1, 1) :
  \Bigl\{
    (v_x, v_v, v_a):fl\times fl\times fl \mid
      \Box(v_x - \frac{v_v^2}{2\times brake}<x_{obs})
      \land \\ \dot{True}
      \land
      \bigl(
        (True \land \Box(v_v \geq 0 \land v_x - \frac{v_v^2}{2\times brake}+0.1\leq x_{obs}))
        \Rightarrow \\
        \Box(v_x - \frac{v_v^2}{2\times brake}<x_{obs})
      \bigr)
  \Bigr\}
}
{}
\end{mathpar}

For this reset, choose $\phi_I\equiv (v_x - \frac{v_v^2}{2\times brake}< x_{obs})$
\begin{mathpar}
\inferrule*[lab=\textsc{T2$_1\equiv$}]
{
  \Gamma,\;
  \mathit{(x,v,a)} : \{(v_x, v_v, v_a):fl\times fl \times fl \mid \Box((v_x-\frac{v_v^2}{2\times brake}<x_{obs}) \land \\ (v_x - \frac{v_v^2}{2\times brake}+0.1 = x_{obs}))\},\;\\
  \mathit{(\dot{x}, \dot{v}, \dot{a})} : \{(w_x, w_v, w_a):fl\times fl\times fl \mid \Box(w_x, w_v, w_a)=(v, a, 0)\}
  \;\vdash\;\\
  (\mathsf{last~}x, \mathsf{last~}v, brake) :
  \Bigl\{
    (v_x, v_v, v_a):fl\times fl\times fl \mid
      \Box(v_x - \frac{v_v^2}{2\times brake}<x_{obs})
      \land \\ \Box v_v- \frac{v_v v_a}{b}\leq 0
      \land
      \bigl(
        (\Box(v_x - \frac{v_v^2}{2\times brake}< x_{obs}) \land True)
        \Rightarrow \\
        \Box(v_x - \frac{v_v^2}{2\times brake}<x_{obs})
      \bigr)
  \Bigr\}
}{}
\end{mathpar}

Choose the same $\phi_I$ for this reset
\begin{mathpar}
\inferrule*[lab=\textsc{T2$_2\equiv$}]
{
  \Gamma,\;
  \mathit{(x,v,a)} : \{(v_x, v_v, v_a):fl\times fl \times fl \mid \\ \Box((v_x-\frac{v_v^2}{2\times brake}<x_{obs}) \land (v_v=0))\},\;\\
  \mathit{(\dot{x}, \dot{v}, \dot{a})} : \{(w_x, w_v, w_a):fl\times fl\times fl \mid \Box(w_x, w_v, w_a)=(v, a, 0)\}
  \;\vdash\;\\
  (\mathsf{last~}x, \mathsf{last~}v, 0) :
  \Bigl\{
    (v_x, v_v, v_a):fl\times fl\times fl \mid
      \Box(v_x - \frac{v_v^2}{2\times brake}<x_{obs})
      \land \\ \Box v_v- \frac{v_v v_a}{b}\leq 0
      \land
      \bigl(
        (\Box(v_x - \frac{v_v^2}{2\times brake}< x_{obs}) \land True)
        \Rightarrow
        \\ \Box(v_x - \frac{v_v^2}{2\times brake}<x_{obs})
      \bigr)
  \Bigr\}
}
{}
\end{mathpar}

\begin{mathpar}
\inferrule*[lab=\textsc{T-DER-FIN}]
{
  \textsc{T1}
  \\
  \textsc{T2$_1$}
  \\
  \textsc{T2$_2$}
  \\
  \textsc{T3}
}
{
  \Gamma \vdash
  prog
  : \{(v_x, v_v, v_a):fl\times fl \times fl \mid \square(v_x-\frac{v_v^2}{2\times brake}<x_{obs})\}
}
\end{mathpar}

\subsection{Additional Examples and Proofs}\label{app:additional}

\subsubsection{Decreasing Sawtooth Function}
We begin with a program which produces a "decreasing sawtooth" signal. The signal's level increases continuously in a linear fashion until it reaches an upper limit ($x=0$), after which it is reset to a lower value which is lower than the upper limit and decreases with each reset. A useful property to prove is that the signal's level never exceeds 0.
\begin{lstlisting}
let rec der x = 1. init -1 reset up(x) -> 
   (let rec xr = -2. fby (xr -. 1.) in xr)
in x
\end{lstlisting}
Ensure: $x\leq 0$ at all times
\\
The invariant, $x\leq 0$ can be proven directly by using the fact that the zero-crossing detector $\mathsf{up}(x)$ is active if the continuous segment begins with $x$ at any value less than 0. This is clearly true for the initial condition, a constant -1. The reset expression can be proven to always be negative using discrete typing rules. Specifically, the (T-LETREC) and (T-FBY) rules from \cite{chen_synchronous_2024} can be used to show that, since $x_r$ was initialized to $-2$, then subtracting $1$ from $x_r$ will always yield a negative number thus allowing us to show $\square (x_r<0)$ and thus not only the invariant is satisfied for the reset but that the guard remains active after the reset (since $x<0$).
\subsubsection{Proof of the Sawtooth Example}~\\

\begin{align*} &\text{T1}\equiv \\ 
  &\Gamma,\;
  x : \{v:\mathsf{fl} \mid \square v\le 0\},\;
  \dot{x} : \{v:\mathsf{fl} \mid \square v = 1.0\}
  \;\vdash\;\\
  &-1.0 :
  \Bigl\{
    v:\mathsf{fl} \;\Bigm|\;
      \square(v \le 0)
      \land \dot{True}
      \land
      \bigl(
        (True \land \square(v\le 0))
        \Rightarrow
        \square(x\le0)
      \bigr)
  \Bigr\}
\end{align*}

\begin{align*}
&\text{T2}\equiv \\
  &\Gamma,\;
  x : \{v:\mathsf{fl} \mid \square v\le0\},\;
  \dot{x} : \{v:\mathsf{fl} \mid \square v = 1.0\},\; x_r:\{v |  \square v< 0\}
  \;\vdash\;\\
  &x_r :
  \Bigl\{
    v:\mathsf{fl} \;\Bigm|\;
      \square(v \le 0)
      \land \dot{True}
      \land
      \bigl(
        (True \land \square(v\le0))
        \Rightarrow
        \square(x\le0)
      \bigr)
  \Bigr\}
\end{align*}

\begin{align*}
&\text{T3}\equiv\\
      &\Gamma,\;
  x : \{v:\mathsf{fl} \mid \Box\,v \le 0\}
  \;\vdash\;
  x : \{v:\mathsf{fl} \mid \Box\,v \le 0\}
\end{align*}

\begin{mathpar}
\inferrule*[lab=\textsc{T-DER-FIN}]
{
  \textsc{T1} \\ \textsc{T2} \\ \textsc{T3}
}
{
  \Gamma \vdash
  \text{let rec der } x : \{v:\mathsf{fl} \mid \Box\,v \le 0\} = \text{1.0 init -1.0 reset}\\
  \;\;|\;\; \text{up}(x) \rightarrow \text{(let rec $x_r$ : $\{v | v<0\}$= -2. fby ($x_r$ -. 1.) in $x_r$)} \\
  \texttt{ in } x
  :
  \{v:\mathsf{fl} \mid \Box\,v \le 0\}
}
\end{mathpar}
\noindent
\subsubsection{Proof that $\square (x_r<0)$:}
Define $\tau\equiv\{v:\mathsf{float}~|~\square(v<0)\}$ for conciseness.
\noindent
$T_{2,4}\equiv\Gamma,x_r:\tau\vdash(x_r-.~1.):\{v:\mathsf{float}~|~\square(v=x_r-1)\}$
\\
(Proven trivially via type synthesis of basic arithmetic operators)\\

\noindent
$T_{2,3}\equiv$
\small{
\begin{prooftree}
    \AxiomC{}
    \RightLabel{(T-CONST)}
    \UnaryInfC{$\Gamma,x_r:\tau\vdash-2.:\{v:\mathsf{float}~|~\square(v=-2)\}$}
    \AxiomC{$T_{2,4}$}
    \RightLabel{(T-FBY)}
    \BinaryInfC{$\Gamma,x_r:\tau\vdash-2.\text{ fby }(x_r-.~1.):\{v:\mathsf{float}~|~\mathsf{hd}(v=-2)\land \bigcirc \square (v=x_r-1)\}$}
\end{prooftree}}
(Trivial application of subtyping applied for $\{v:\mathsf{float}~|~\square(v=-2)\}\preceq\{v:\mathsf{float}~|~(v=-2)\}$
\\
$T_{2,1}\equiv$
\begin{prooftree}
    \AxiomC{$T_{2,3}$}
    \AxiomC{$\{v:\mathsf{float}~|~\mathsf{hd}(v=-2)\land \bigcirc \square (v=x_r-1)\}\preceq \tau$}
    \AxiomC{$\Gamma,x_r:\tau\vdash\tau$}
    \RightLabel{(T-SUB)}
    \TrinaryInfC{$\Gamma,x_r:\tau\vdash-2.\text{ fby }(x_r-.~1.):\tau$}
\end{prooftree}
(Subtyping obligation is proven by noting that $\forall x_r:\mathsf{float}~.~x_r<0\implies (x_r - 1) < 0$, and that $2<0$).
\\
\noindent
$T_{2,2}\equiv$
\begin{prooftree}
    \AxiomC{$(\Gamma,x_r:\tau)(x_r)=\tau$}
    \RightLabel{(T-VAR)}
    \UnaryInfC{$\Gamma,x_r:\tau\vdash x_r:\tau$}
\end{prooftree}

\begin{prooftree}
    \AxiomC{$\Gamma,x_r:\tau\vdash \tau$}
    \AxiomC{$T_{2,1}$}
    \AxiomC{$T_{2,2}$}
    \RightLabel{(T-LETREC)}
    \TrinaryInfC{$\Gamma \vdash \text{let rec }x_r:\tau=-2.\text{ fby }(x_r-.~1.) \text{ in } x_r:\tau$}
\end{prooftree}

\subsubsection{Bouncing Ball}
The bouncing ball is a classic hybrid system which, while not necessarily controlled by a software component, nevertheless demonstrates a system containing both continuous dynamics and a discrete event which resets the system's state. Here, the ball is released from a positive height and allowed to free-fall until it reaches the ground, at which point it bounces up with some energy lost during the bounce. We verify that the ball never travels through the ground, and that, obeying the laws of physics, does not bounce higher than its initial height. Safety of this system is proven by using both the reset condition and differential reasoning similarly to \cite{platzer2018lfcps} to ensure the ball obeys conservation of energy.
\begin{lstlisting}
let y0 = 10.0; v0 = 0.0; g = 9.81
let rec der (y, v) = 
    (v init y0,
    -.g init v0) reset up(-.y) -> (last y, -0.8 *. last v)
 in (y, v)
\end{lstlisting}
Ensure: $0 \leq y \leq y_0$
\\
The proof applies rule (T-DER-FIN) to show that the ball’s height never becomes negative and remains within the specified bounds. The bouncing ball can be proven safe by splitting the invariant; $y\geq 0$ can be proven entirely using the fact that $\mathsf{up}(-y)$ is always active. Similarly to \cite{platzer2018lfcps}, $y\leq y_0$ can be proven using the invariant $2gy\leq 2gy_0-v^2$ which is just a reformulation of the potential and kinetic energies of the ball. We verify the invariant $y\geq 0 \land 2gy \leq 2gy_0-v^2$. The invariant combines a geometric safety condition on the height with an energy bound relating height and velocity. Since $Active(y_0, \{y\})$ lets us directly prove $y\geq 0$, we select $\phi_I\equiv 2gy\leq 2gy_0-v^2$, and subsequently note that $\phi_i'\equiv -2gv\leq -2gv$ During continuous evolution under gravity, the energy component of the invariant is preserved by the trivially true differential invariant $\phi_i'$, while the zero-crossing guard prevents the height from crossing below zero. When the guard triggers at ground contact, the reset maps the state to a strictly lower-energy configuration with upward velocity (thus $(-0.8v)^2 \leq v^2\leq 2gy_0$), immediately re-establishing the invariant and providing a safe initial state for the next segment. By induction over alternating continuous evolutions and resets, the invariant holds globally.

\subsubsection{Proof of the Bouncing Ball Example}

\begin{align*}
&{T3}\equiv \\
&{\Gamma,\;
  (y,v) : \{(u, w) : fl \times fl \mid \Box(0.0 \le u \le 10.0 \land 2gy\le2gy_0 - v^2)\}
  \;\vdash\;
  }\\&{  y : \{u : fl \mid \Box(0.0 \le u \le 10.0)\}
}
\end{align*}

\begin{align*}
&{T1}\equiv \\
  &\Gamma,\;
  y : \{u \mid \square u \ge 0\},\;
  \dot{y} : \{u \mid\square u = v\},\;
  \dot{v} : \{w \mid\square w = -9.81\}
  \;\vdash\;\\
  &(10.0, 0.0) :
  \Bigl\{
    (u, w) \mid
      \square(0.0 \le u \le 10.0 \land 2gu\le2gy_0 - w^2) \\
      &\land {\square((y\le 10.0) \land (2 g y \le 2 g y_0 - v^2))}'
      \land\\
      &\bigl(
        \square(((y \le 10.0) \land (2 g y \le 2 g y_0 - v^2)) \land \square(u \ge 0.0))
        \Rightarrow
        \square(0.0 \le y \le 10.0)
      \bigr)
  \Bigr\}
\end{align*}

\begin{align*}
&\text{T2}\equiv
  \\&\Gamma,\;
  y : \{u \mid \square (u \ge 0) \land u = 0\},\;
  \dot{y} : \{u \mid \square u = v\}
  \;\vdash\;\\&
  (0.0, -0.8 \cdot v) :
  \Bigl\{
    (u, w) \mid
      \square(0.0 \le u \le 10.0 \land 2gu\le2gy_0 - w^2) \\&
      \land {\square((0.0\le y) \land (y\le 10.0) \land (2 g y \le 2 g y_0 - v^2))}'\\&
      \land
      \bigl(
        \square(((0.0\le y) \land (y\le 10.0) \land (2 g y \le 2 g y_0 - v^2)) \land True)
        \Rightarrow \\&
        \square(0.0 \le u \le 10.0 \land 2gy\le2gy_0 - v^2)
      \bigr)
  \Bigr\}
\end{align*}

\begin{mathpar}
\inferrule*[lab=\textsc{T-DER-FIN}]
{
  \textsc{T1}
  \\
  \textsc{T2}
  \\
  \textsc{T3}
}
{
  \Gamma \vdash
  \text{let rec}\
    \text{der (y,v)} : \{(u, w) : fl \times fl \mid \Box(0.0 \le u \le 10.0 \land 2gy\le2gy_0 - v^2)\} \\ = (v, -9.81) \text{ init } (10.0, 0.0) \\ \text{ up(-y) $\to$ (last y, -0.8 * last v)}
\
  \;\texttt{in}\; y
  :
  \{u : fl \mid \Box(0.0 \le u \le 10.0)\}
}
\end{mathpar}
\end{document}